\documentclass[lettersize,journal]{IEEEtran}
\usepackage{amsmath,amsfonts}
\usepackage{algorithmic}
\usepackage{array}
\usepackage[caption=false,font=normalsize,labelfont=sf,textfont=sf]{subfig}
\usepackage{enumitem}
\usepackage{textcomp}
\usepackage{stfloats}
\usepackage{makecell}
\usepackage{diagbox}
\usepackage{multirow}
\usepackage{url}
\usepackage{verbatim}
\usepackage{graphicx}
\def\BibTeX{{\rm B\kern-.05em{\sc i\kern-.025em b}\kern-.08em
    T\kern-.1667em\lower.7ex\hbox{E}\kern-.125emX}}
\usepackage{balance}
\usepackage{verbatim}
\usepackage{amsfonts}
\usepackage{amssymb}
\usepackage{stfloats}
\usepackage{cite}
\usepackage{soul}
\usepackage{setspace}
\usepackage{graphicx}
\usepackage{psfrag}
\usepackage{amsmath}
\usepackage{array}
\usepackage{epstopdf}
\usepackage{authblk}
\usepackage{graphicx}
\usepackage{amsthm}
\usepackage{lipsum}
\usepackage{verbatim}
\usepackage{authblk}
\usepackage{mathtools}
\usepackage{cuted}
\usepackage[lined,boxed,ruled]{algorithm2e}
\usepackage{booktabs}
\usepackage[font=small]{caption} 

\usepackage{ifthen}

\usepackage[usenames]{color}

\newtheorem{Proposition}{Proposition}
\newtheorem{Remark}{Remark}

\newboolean{showaircomp}
\setboolean{showaircomp}{False} 

\begin{document}

\title{Joint Optimization of Pilot Length, Pilot Assignment, and Power Allocation for Cell-free MIMO Systems with Graph Neural Networks}
\author{ 
{Yao Peng, Tingting Liu, and Chenyang Yang}
\thanks{Y. Peng, T. Liu, and C. Yang are with the School of Electronics and Information Engineering, Beihang University, Beijing, China (Email: \{pengyyao, ttliu, cyyang\}@buaa.edu.cn).
Corresponding author: T. Liu.

This work has been submitted to the IEEE for possible publication. Copyright may be transferred without notice, after which this version may no longer be accessible.
}
}%

\markboth{Journal of \LaTeX\ Class Files,~Vol.~xx, No.~xx, xx~2025}%
{Joint Optimization of Pilot Length, Pilot Assignment, and Power Allocation for Cell-free MIMO Systems with Graph Neural Networks}

\maketitle

\begin{abstract}
In user-centric cell-free multi-antenna systems, pilot contamination degrades spectral efficiency (SE) severely. To mitigate pilot contamination, existing works jointly optimize pilot assignment and power allocation by assuming fixed pilot length, which fail to balance pilot overhead against the contamination. To maximize net-SE, we jointly optimize pilot length, pilot assignment, and power allocation with deep learning. Since the pilot length is a variable, the size of pilot assignment matrix is unknown during the optimization. To cope with the challenge, we design size-generalizable graph neural networks (GNNs). We prove that pilot assignment policy is a one-to-many mapping, and improperly designed GNNs cannot learn the optimal policy. We tackle this issue by introducing feature enhancement. To improve learning performance, we design a contamination-aware attention mechanism for the GNNs. Given that pilot assignment and power allocation respectively depend on large- and small-scale channels, we develop a dual-timescale GNN framework to explore the potential. To reduce inference time, a single-timescale GNN is also designed. Simulation results show that the designed GNNs outperform existing methods in terms of net-SE, training complexity, and inference time, and can be well generalized across problem scales and channels.

\end{abstract}

\begin{IEEEkeywords}
	Cell-free, joint pilot length, pilot assignment, and power allocation, graph neural network.
\end{IEEEkeywords}

%
\section{Introduction}
User-centric cell-free multiple-input multiple-output (CF-MIMO) has emerged as a promising technology for future wireless networks. By enabling nearby access points (APs) to serve each user equipment (UE) cooperatively \cite{2020B5G}, the signal-to-interference-plus-noise ratio (SINR) can be increased, providing substantial gains in both the overall system spectral efficiency (SE) and user experience \cite{2022UCSurvey-Ammar}.

Since assigning a unique pilot sequence (PS) to each UE incurs large pilot overhead, pilot reuse is inevitable. However, this introduces pilot contamination, which severely degrades the SE gains \cite{2022UCSurvey-Ammar}. To improve the SE of the CF-MIMO systems under pilot contamination, pilot assignment and power allocation have been jointly optimized \cite{2024JPCPA-CF-Ren,2024JPCPA-CF-Khan,JPAPC2025}. To balance the pilot overhead with pilot contamination, the pilot length needs to be optimized. Specifically, longer PSs reduce pilot contamination by allowing more orthogonal PSs, yet they also reduce the resources for data transmission. 
Therefore, to maximize the net-SE, which accounts for the impact of pilot overhead, the pilot length should be jointly optimized together with pilot assignment and power allocation.

\subsection{Related Works}
We review existing works in the literature related to power allocation optimization, pilot length and pilot assignment optimization, and pilot assignment and power allocation joint optimization.

\subsubsection{Optimizing power allocation}
Since power allocation problems are usually non-convex, recent works have employed deep neural networks (DNNs) to learn power allocation policies, which are the mappings from environmental states to power allocation \cite{2023PAFNN-Fab,2024PAGNN-Mishra,CSIPower2024,2024GC_PY}, to achieve near-optimal performance at low inference latency.
Compared with fully-connected neural networks (FNNs) \cite{2023PAFNN-Fab}, graph neural networks (GNNs) \cite{2024PAGNN-Mishra,CSIPower2024,2024GC_PY} offer significant advantages in terms of training efficiency and size generalizability.
This is because GNNs can leverage the permutation properties of optimal policies, a form of prior knowledge in wireless communication problems \cite{LSJ2023MGNN,2024GNN_or_CNN}.

However, perfect channel information was assumed known in \cite{CSIPower2024,2024GC_PY}. Although the channel estimation errors were considered in \cite{2024PAGNN-Mishra,2023PAFNN-Fab}, the pilot length or pilot assignment was not optimized, leaving the potential in performance gain from optimizing power allocation unexploited.

\subsubsection{Optimizing pilot length and assignment}
The optimization of pilot resources aims to improve net-SE by trading off the impacts of pilot overhead and channel estimation accuracy.

In contamination-free scenarios (e.g., when the number of UEs is much smaller than the number of available orthogonal PSs), the primary focus is on optimizing pilot length to balance pilot overhead and channel estimation errors \cite{Cheng2017,Zhou2023}. In such cases,  each UE is assigned a unique orthogonal PS, and not all available PSs need to be utilized.
In contamination-limited scenarios (e.g., cell-free systems where the number of UEs exceeds the number of available orthogonal PSs), all PSs must be assigned, and the key challenge shifts to efficiently reusing pilots for mitigating the contamination. Since pilot length determines the number of orthogonal PSs, it is essential to select an appropriate pilot length and optimize the corresponding pilot assignment under different levels of pilot contamination.

Several studies have optimized pilot assignment given a predetermined pilot length \cite{cellfreeVersus2017,2020Tabu,2020GraphColoring,2021KCut,ClusteringPA2025,2022pilotRDL-Rahmani}. Search-based algorithms \cite{cellfreeVersus2017,2020Tabu} start with an initial assignment and iteratively reassign PSs based on specific rules. Graph-based methods \cite{2020GraphColoring,2021KCut,ClusteringPA2025} first construct interference graphs, where vertices represent UEs and the weights of edges quantify the potential contamination between UEs. Then, pilot assignment is performed using graph coloring \cite{2020GraphColoring} or Max K-cut algorithms \cite{2021KCut,ClusteringPA2025}. For distributed implementation, multi-agent reinforcement learning (MARL) is adopted in \cite{2022pilotRDL-Rahmani}.

Recently, the pilot length and assignment have been jointly optimized in \cite{Peng2023,Wu2023}. In \cite{Peng2023}, the maximal number of UEs whose quality-of-service (QoS) constraints can be satisfied was found by iteratively refining an interference graph based on the current QoS of UEs, updating the assignment with Dsatur coloring, and setting the pilot length to the number of assigned PSs.
In \cite{Wu2023}, net-SE was maximized via an evolutionary algorithm. Despite their effectiveness, both methods exhibit high computational complexity.

\subsubsection{Joint optimization of pilot assignment and power allocation}
For a predetermined pilot length, pilot assignment and power allocation were jointly optimized in \cite{2024JPCPA-CF-Ren,2024JPCPA-CF-Khan,JPAPC2025}. It is worthy noting that pilot assignment should be optimized on a large timescale (e.g., seconds), because it is closely tied to AP-UE associations, and frequent adjustments introduce significant signalling overhead. In contrast, power allocation is typically optimized on a much smaller timescale (e.g., milliseconds) to pursue optimal performance. This difference renders the joint optimization a dual-timescale (DTS) problem. Yet single-timescale (STS) problems were formulated in \cite{2024JPCPA-CF-Ren,2024JPCPA-CF-Khan,JPAPC2025}, which optimize both variables only on the large timescale.

Specifically, in \cite{2024JPCPA-CF-Ren}, an SE maximization problem was formulated as a series of Max K-cut problems, and a sequential algorithm was proposed to solve them. To enable real-time implementation, the learning-based approach was adopted in \cite{2024JPCPA-CF-Khan,JPAPC2025}. In \cite{2024JPCPA-CF-Khan}, an FNN-based method was proposed, where one FNN handles power allocation while multiple parallel FNNs
are used to assign PSs to each UE individually. In \cite{JPAPC2025}, MARL was employed, where each UE acts as an agent that selects a PS and optimizes its own transmit power. However, these DNNs suffer from high training complexity.

\subsection{Motivation and Major Contributions}
In this paper, we jointly optimize pilot length, pilot assignment, and power allocation in both DTS and STS with deep learning. Since the pilot length is an optimization variable, the number of PSs is unknown. Consequently, the pilot assignment matrix, which represents the assignment of PSs to UEs, has an unknown size. For making the joint optimization, the pilot assignment matrix should be an output of the DNN. However, any DNN must have a predetermined output size.

To resolve the challenge due to the size uncertainty, we optimize another pilot assignment matrix with a predetermined size, where each row corresponds to one PS, and the total number of rows is the maximal number of assignable PSs. This maximal number of PSs can be set as the number of UEs, which is a conservative value to avoid pilot contamination. When fewer PSs are actually assigned than the maximum after the optimization, the rows corresponding to unassigned PSs are set to zero. 

Although the size of the pilot assignment matrix to be optimized is predetermined, it still depends on the number of UEs, which is dynamic in practice. Therefore, the DNN for the joint optimization must be generalizable to different numbers of UEs.

Previous works have confirmed that GNNs exhibit size generalizability \cite{LSJ2023MGNN,2024GNN_or_CNN}. This capability is achieved by designing parameter-sharing schemes that harness the permutation properties of optimal policies, thereby making the trainable parameters independent of the problem scales. This observation motivates us to develop GNN-based frameworks for joint optimization problems.

The major contributions are summarized as follows.
\begin{itemize}
	\item We design GNNs that jointly optimize pilot length, pilot assignment, and power allocation to maximize net-SE.
 Consequently, the proposed GNNs outperform existing methods.
 By leveraging the permutation properties of optimal policies, the designed GNNs can be well generalized to varying numbers of UEs, APs, and antennas.
	\item We prove that the pilot assignment policies are one-to-many mappings. When the GNNs for pilot assignment are trained using only environmental states as input, they often produce infeasible solutions or suffer from severe pilot contamination. To convert the pilot assignment policies to be learned into one-to-one mappings, we resort to feature enhancement. Furthermore, to improve the learning performance, we design an attention mechanism that can measure potential pilot contamination for the GNNs for pilot assignment.
	\item To explore the potential in boosting the net-SE, we develop a DTS framework comprising two GNNs for pilot assignment and power allocation, each operating at different timescales. To strike a balance between system performance and inference complexity, we also develop an STS framework including a GNN for joint optimization on the large timescale.
   \item Simulations show that the proposed GNNs outperform numerical and learning-based baselines in terms of net-SE, training complexity, inference latency, and can be well generalized to different system settings. Moreover, in the scenario with a large number of antennas, the STS framework incurs almost no loss in net-SE compared to the DTS framework, while having lower space complexity and faster inference.
\end{itemize}

The remainder of this paper is organized as follows. Section \ref{sec:System-Model} introduces the system model. Section \ref{sec:Problem-Policy} introduces the joint optimization problems, the optimal policies to be learned, and the properties of these policies. Section \ref{sec:GNN} introduces the designed GNNs. Sections \ref{sec:Result} and \ref{sec:Conclusion} provide simulation results and conclusions, respectively.

\emph{Notations:}
$\mathbb R$ and ${\mathbb C}$ denote the sets of real and complex numbers, respectively. ${\bf A}=[a_{lk}]_{L\times K}$ denotes a matrix, where $a_{lk}$ is the element in the $l$th row and the $k$th column, $l=1,\cdots, L$ and $k=1,\cdots,K$. ${\bf 1}_{K\times 1}$ denotes a column vector with all the $K$ elements equal to one. ${\bf I}_N$ denotes an $N\times N$ identity matrix. ${\bf\Pi}$ denotes a permutation matrix.
$(\cdot)^{\rm T}$ and $(\cdot)^{\rm H}$ denote the transpose and Hermitian transpose, respectively. ${\sf Tr}(\cdot)$ denotes the trace of a matrix.
$\Vert\cdot\Vert$ denotes the Frobenius norm. $|\cdot|$ denotes the magnitude of a complex number or the cardinality of a set.
$\otimes$ and $\odot$ denote the Kronecker and Hadamard products, respectively.

%

\section{System Model}\label{sec:System-Model}
Consider a user-centric time-division duplex (TDD) CF-MIMO system, where $M$ APs, each equipped with $N$ antennas, serve $K$ single-antenna UEs. All APs are connected to a central unit (CU), and each UE is associated with nearby APs.

\subsection{Frame Structure and Channel Model}
The structure of a frame is shown in Fig. \ref{fig:vary_CSI}. Each frame consists of $N_{\sf T}$ subframes. Each subframe is divided into $\tau_{\sf c}$ time slots, where $\tau_{\sf p}$ slots are used for PS transmission, and the remaining $\tau_{\sf c}-\tau_{\sf p}$ slots are used for data transmission.

\vspace{-3mm}
\begin{figure}[!htp]
	\centering
	\includegraphics[scale=0.5]{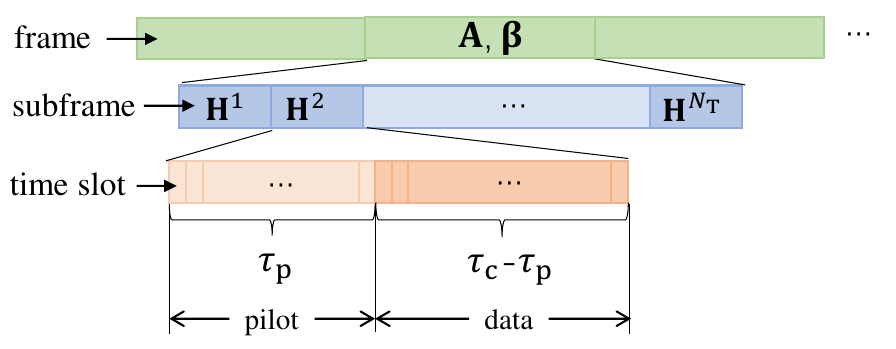}
	\caption{Frame structure of the TDD system.} \label{fig:vary_CSI}
\end{figure}
\vspace{-3mm}

We assume that the large-scale fading (LSF) channels remain unchanged in a frame, and the small-scale fading (SSF) channels are independent and identically distributed (i.i.d.) across subframes and UEs. Thus, $\tau_{\sf c}$ can reflect the channel coherence time.

The user association depends on the LSF channels and thus remains unchanged within a frame. Denote
 ${\bf A}=[a_{mk}]_{M\times K}\in\{0,1\}^{M\times K}$ as the association matrix, where $a_{mk}=1$ if the $k$th UE (denoted as ${\sf UE}_k$) is associated with the $m$th AP (denoted as ${\sf AP}_m$), and $a_{mk}=0$ otherwise.

The channel vector from ${\sf AP}_m$ to ${\sf UE}_k$ in the $t$th subframe is $\sqrt{\beta_{mk}}{\bf h}_{mk}^{t}\in{\mathbb C}^{N\times 1}$, where $\beta_{mk}$ is the LSF channel gain, and ${\bf h}_{mk}^{t}$ is the SSF channel vector. As in \cite{foundation_CF_2021}, we assume that the channel correlation matrix ${\bf R}_{mk}\triangleq{\mathbb E}_{{\bf h}_{mk}^{t}}\{(\sqrt{\beta_{mk}}{\bf h}_{mk}^{t})(\sqrt{\beta_{mk}}{\bf h}_{mk}^{t})^{\sf H}\}\in{\mathbb C}^{N\times N}$ is known, where the expectation is taken over ${\bf h}_{mk}^{ t}$. Hence, $\beta_{mk}\triangleq{\sf Tr}({\bf R}_{mk})/N$ is also known.

Denote ${\boldsymbol\beta}=[\beta_{mk}]_{M\times K}\in{\mathbb R}^{M\times K}$ as the LSF channel matrix, and
${\bf H}^{t}=[{\bf h}_{mk}^{t}]_{M\times K}\in{\mathbb C}^{MN\times K}$ as the SSF channel matrix in the $t$th subframe.

\subsection{Pilot Assignment}\label{sec:system_PA}
At the beginning of a frame, the CU assigns $\tau_{\sf p}$ mutually orthogonal PSs, each of length $\tau_{\sf p}$, to the UEs.

Denote ${\bf x}_k=[x_{1k},\cdots,x_{\tau_{\sf p}k}]^{\sf T}$ as the pilot assignment vector for ${\sf UE}_k$, where $x_{gk}=1$ if the $g$th PS (denoted as ${\sf PS}_g$) is assigned to ${\sf UE}_k$, $x_{gk}=0$ otherwise. ${\bf X}_{\sf o}=[{\bf x}_1,\cdots,{\bf x}_{\tau_{\sf p}}]\in\{0,1\}^{\tau_{\sf p}\times K}$ is the pilot assignment matrix for $K$ UEs. ${\bf X}_{\sf o}^{\sf T}\cdot{\bf 1}_{\tau_{{\sf p}}\times 1}={\bf 1}_{K\times 1}$, since each UE is assigned only one PS.

\subsection{Channel Estimation and Beamforming} \label{sec:channel-estimation}
In the $t$th subframe, the received pilot signal at ${\sf AP}_m$ for estimating the channel of ${\sf UE}_k$ is
\begin{equation} \label{eq:received-pilot}
{{\bf{y}}_{mk}^{t}} = \sum\nolimits_{i=1}^K({\bf x}_k^{\sf T}{\bf x}_i)\sqrt{p^{\sf ul}\tau_{\sf p}}\sqrt{\beta_{mi}}\mathbf{h}_{mi}^{t}+{\bf n}_{mk}^{t},
\end{equation}
where ${\bf x}_k^{\sf T}{\bf x}_i=1$ if ${\sf UE}_k$ and ${\sf UE}_i$ are assigned with the same PS, ${{\bf x}_k^{\sf T}}{\bf x}_i=0$ otherwise, $p^{\sf ul}$ is the uplink transmit power of each UE, ${\bf n}_{mk}^{t}\sim{\mathcal {CN}}(0,\sigma_{{\sf AP},m}^2{\bf I}_N)$, and $\sigma_{{\sf AP},m}^2$ is the noise power  at ${\sf AP}_m$.

To reduce the overhead for exchanging channels among APs, we adopt the distributed beamforming discussed in \cite{foundation_CF_2021}, where each AP computes its beamforming vectors independently, using the channel vectors of its associated UEs (i.e., the local channels) and the known LSF channel gains of unassociated UEs. In particular, ${\sf AP}_m$ computes a beamforming vector for every associated UE, say ${\bf v}_{m_k}^{t}\in{\mathbb C}^{N\times 1}$ for ${\sf UE}_k$, $k\in{\mathcal A}_m$, where ${\mathcal A}_m=\{k|a_{mk}=1\}$ is a set containing the indices of UEs associated with ${\sf AP}_m$, and $\Vert{\bf v}_{m_k}^{t}\Vert=1$.

Consequently, only the local channels need to be estimated.
Using the minimum mean square error (MMSE) estimator in \cite{foundation_CF_2021}, the estimated channel vector from ${\sf AP}_m$ to ${\sf UE}_k$ in the $t$th subframe is given by
\begin{equation} \label{eq:estimated-channel}
\begin{aligned}
	{\hat {\bf{h}}_{mk}^{t}}\!\!\!=\!\! {\sqrt {{p^{\sf ul}}} }{\bf R}_{mk}\big(\!\sum\nolimits_{i=1}^K\!\!({{\bf x}_k}^{\!\!\sf T}{\bf x}_i)p^{\sf ul}{\bf R}_{mi}\!\!+\!\!\sigma_{{\sf AP},m}^2{\bf I}_N\!\big)^{\!\!-1}\!{\bf{y}}_{mk}^{t}.
\end{aligned}
\end{equation}

\subsection{Performance Metric}\label{sec:DL_transmission}
In the $t$th subframe, the SINR of ${\sf UE}_k$ is
\begin{equation} \label{eq:SINR}
	\gamma_{k}^{t}=\frac{\left|\sum_{m\in{\mathcal S}_k}\sqrt{p_{m_k}^{t}}g_{m_kk}^{t}\right|^2}{\sum_{i\neq k}\left|\sum_{m\in{\mathcal S}_i} \sqrt{p_{m_i}^{t}}g_{m_ik}^{t}\right|^2+\sigma_{{\sf UE},k}^2},
\end{equation}
where $g_{m_ik}^{t}\triangleq ({\beta}_{mk}{\bf h}_{mk}^{t})^{\rm H}{\bf v}_{m_i}^{t}$ is the equivalent channel coefficient from ${\sf AP}_m$ to ${\sf UE}_k$ after beamforming for ${\sf UE}_i$,\footnote{$g_{m_ik}^{t}$ can be regarded as the channel coefficient from the \emph{equivalent antenna} (called AN for short) at ${\sf AP}_m$ serving ${\sf UE}_i$ (denoted ${\sf AN}_{m_i}$) to ${\sf UE}_k$.} ${\mathcal S}_k=\{m|a_{mk}=1\}$ is a set containing the indices of APs serving ${\sf UE}_k$, $\sigma_{{\sf UE},k}^2$ is the noise power at ${\sf UE}_k$, and $p_{m_i}^{t}$ is the transmit data power from ${\sf AP}_m$ to ${\sf UE}_i$.

Denote the power allocation matrix by ${{\bf P}^t}=[p_{m_k}^t]_{M\times K}\in{\mathbb R}^{M\times K}$.
To account for the impact of pilot overhead, the net-SE in the $t$th subframe is defined as
\begin{equation} \label{eq:net-SE}
\begin{aligned}
	\eta^t&={\sf SE}({\tau_{\sf p}},{\bf X}_{\sf o},{\bf P}^t,{\bf A},{\boldsymbol{\beta}},{\bf H}^t)&
    \\&\triangleq\left(1-{\tau_{\sf p}}/{\tau_{\sf c}}\right)\sum\nolimits_{k=1}^K\log_2(1+\gamma_{k}^{t}),
    \end{aligned}
\end{equation}
where ${\sf SE}(\cdot)$ is a function.
The average net-SE in a frame is
 \begin{equation}\label{eq:Avg-net-SE}
 	{\bar \eta}=\frac{1}{N_{\sf T}}\sum\nolimits_{t=1}^{N_{\sf T}}\eta^t.
 \end{equation}

Since SSF channels in different subframes are i.i.d., we have $\bar \eta \approx {\mathbb E}_{{\bf H}^t}\{\eta^t\}$, where the expectation is taken over the SSF channels, and the approximation is accurate when $N_{\sf T}$ is sufficiently large.

\section{Joint Pilot Length, Pilot Assignment, and Power Allocation} \label{sec:Problem-Policy}
In this section, we formulate the joint optimization problems in DTS and STS. We proceed to present the corresponding policies and analyze their permutation properties. Finally, we show that the optimal pilot assignment policies are one-to-many mappings and discuss their impact.

\subsection{Problems Formulation for Joint Optimization}
We aim to maximize \(\mathbb{E}_{\mathbf{H}^t} \{\eta^t\}\) by jointly optimizing the pilot length \(\tau_{\mathsf{p}}\), the pilot assignment matrix \({\mathbf{X}}_{\sf o}\), and the power allocation matrix \(\mathbf{P}^t\).

Since \(\tau_{\sf p}\) itself is an optimization variable, the size of \({\mathbf{X}}_{\sf o} \in \{0,1\}^{\tau_{\sf p} \times K}\) is unknown a priori. To address the issue of size uncertainty, we introduce another pilot assignment matrix \(\mathbf{X}=[x_{gk}]_{K\times K} \in \{0,1\}^{K \times K}\), which is capable of representing all feasible pilot assignments. When $\tau_{\sf p} = K$, all UEs are assigned with orthogonal PSs, no matter if pilot contamination exists among them.

Both \(\tau_{\sf p}\) and \(\mathbf{X}_{\sf o}\) can be derived from \(\mathbf{X}\) as follows.

Each row of \(\mathbf{X}\) corresponds to a potential PS. When each UE is assigned exactly one PS, there are $K$ potential PSs. However, since the PSs will be reused for maximizing the net-SE, only a subset of these PSs is actually assigned while the rest of PSs are unused. Consequently, the rows of \(\mathbf{X}\) that are entirely zero indicate unused PSs. In the pilot contamination-limited scenario, the optimal pilot length equals the number of PSs actually assigned \cite{Peng2023}. Therefore, \(\tau_{\sf p}\) is determined by counting the number of non-all-zero rows in \(\mathbf{X}\), and \(\mathbf{X}_{\sf o}\) is constructed by removing the all-zero rows from \(\mathbf{X}\).

An example is illustrated in Fig. \ref{fig:pilot_assignment_matrix}. In this case, \(\mathsf{PS}_1\) is assigned to \(\mathsf{UE}_1\) and \(\mathsf{UE}_3\), \(\mathsf{PS}_2\) is assigned to \(\mathsf{UE}_2\), and \(\mathsf{PS}_3\) is unused. Thus, \(\tau_{\sf p} = 2\), and \(\mathbf{X}_{\sf o}\) is obtained by removing the third row of \(\mathbf{X}\).
\vspace{-2mm}
\begin{figure}[!htp]
	\centering
	\includegraphics[scale=0.53]{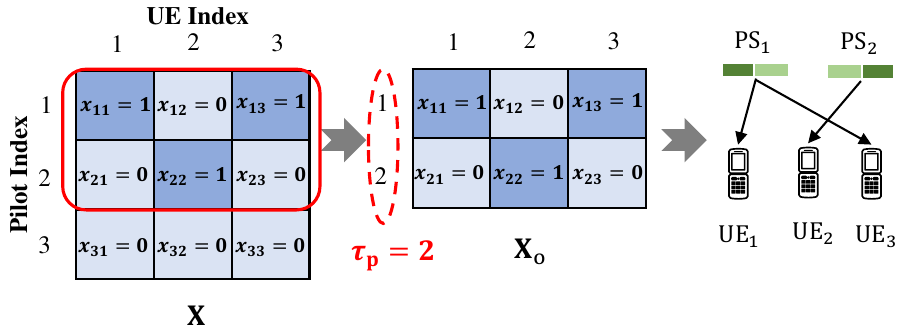}
	\vspace{-2mm}
	\caption{$\mathbf{X}$ with $K=3$ and the derived $\tau_{\sf p}$ and $\mathbf{X}_{\sf o}$.} \label{fig:pilot_assignment_matrix}
	\vspace{-4mm}
\end{figure}

The original joint optimization problem over three variables is thus equivalent to jointly optimizing \(\mathbf{X}\) and \(\mathbf{P}^t\) to maximize \(\mathbb{E}_{\mathbf{H}^t} \{\eta^t \mid \tau_{\mathsf{p}} = \psi(\mathbf{X}) \}\), where $\psi(\cdot)$ is a function for computing the number of assigned PSs. The expression of $\psi(\cdot)$ is
\vspace{-2mm}
\begin{equation} \label{eq:NumberPS}
	\psi(\mathbf{X}) = \sum_{g=1}^{K} \left(1 - \prod_{k=1}^K (1 - x_{gk}) \right),
\end{equation}
where \(\prod_{k=1}^K (1 - x_{gk}) = 1\) if the $g$th row of ${\bf X}$ is entirely zero (i.e., $\mathsf{PS}_g$ is not assigned to any UE), and \(\prod_{k=1}^K (1 - x_{gk}) = 0\) if $\mathsf{PS}_g$ is assigned to at least one UE.


\subsubsection{DTS problem} We formulate a DTS problem to jointly optimize the pilot assignment matrix in a frame and the power allocation matrix in each subframe, which is
\begin{subequations}\label{eq:problem}
	\begin{align}
		\begin{split} \label{eq:objective-function}
            \mathop {\sf max }\limits_{{\bf{X}},{\bf{P}}^t}&\mathbb{E}_{\mathbf{H}^t} \{\eta^t | \tau_{\mathsf{p}} = \psi(\mathbf{X}) \}
		\end{split}\\
		\begin{split} \label{eq:cons-1}
			s.t.~&{\bf X}^{\sf T}\cdot{\bf 1}_{K\times 1}={\bf 1}_{K\times 1},~{x_{gk}} \in \{ 0,1\},\forall g,k,
		\end{split}\\
		\begin{split}\label{eq:cons-2}
			~&({\bf A}\odot{\bf P}^t)\cdot{\bf 1}_{K\times 1}\preceq P_{\sf max}{\bf 1}_{K\times 1},~{\bf P}^t \succeq 0.
		\end{split}
	\end{align}
\end{subequations}
Constraint \eqref{eq:cons-1} ensures that each UE is assigned one PS, and constraint \eqref{eq:cons-2} restricts the total transmit power of each AP not exceeding $P_{\sf max}$.

The DTS problem can be solved by a bi-level optimization method \cite{2025MTS}. Specifically, the DTS problem can be equivalently transformed into the following two subproblems,
\begin{subequations} \label{eq:subproblems}
	\begin{align}
		\!&{\textbf{high-level subproblem:}} \nonumber\\
		\begin{split} \label{eq:high-level}
                \!& \mathop {\sf max }\limits_{{\bf{X}}}~ \mathbb{E}_{\mathbf{H}^t} \{\eta^t | \tau_{\mathsf{p}} = \psi(\mathbf{X}) \},
		   s.t.~\eqref{eq:cons-1},\mathbf{P}^t \!=\! P^*({\bf X}, {\mathbf A}, {\boldsymbol\beta},{\bf H}^t),
	    \end{split}\\
		\!&{\textbf {low-level subproblem:}}\nonumber\\
		\begin{split} \label{eq:low-level}
			\!& \mathop{\sf max }\limits_{P({\bf X}, {\mathbf A}, {\boldsymbol\beta},{\bf H}^t)}~\{\eta^t|{\bf X},\tau_{\mathsf{p}} = \psi(\mathbf{X})\}, ~s.t.~\eqref{eq:cons-2}.
	    \end{split}
	\end{align}
\end{subequations}

The low-level subproblem optimizes a power allocation function $\mathbf{P}^t = P(\mathbf{X}, {\mathbf A}, {\boldsymbol\beta},\mathbf{{H}}^t)$  in each subframe, given the pilot assignment. The high-level subproblem optimizes the pilot assignment matrix in a frame, given the optimal power allocation function $P^*(\mathbf{X}, {\mathbf A}, {\boldsymbol\beta}, \mathbf{{H}}^t)$.

\subsubsection{STS problem}
To strike a balance between the net-SE and inference complexity, we optimize both pilot assignment and power allocation in a frame, i.e.,
\begin{subequations}\label{eq:problem-LSF}
	\begin{align}
		\begin{split} \label{eq:objective-function-LSF}
			\mathop {\sf max }\limits_{{\bf{X}},{\bf{P}}} ~&{\mathbb E}_{{\bf H}^t}\{\eta^t|\tau_{\mathsf{p}} = \psi(\mathbf{X})\},
		\end{split}\\
		\begin{split}
			s.t.~&\eqref{eq:cons-1},
		\end{split}\label{eq:cons-LSF-1} \\
		\begin{split}
			~&({\bf A}\odot{\bf P})\cdot{\bf 1}_{K\times 1}\preceq P_{\sf max}{\bf 1}_{K\times 1},~{\bf P} \succeq 0
		\end{split}\label{eq:cons-LSF}
	\end{align}
\end{subequations}
where ${\bf P}\in{\mathbb R}^{M\times K}$ is the power allocation matrix in the frame.

The STS optimization suffers a net-SE loss due to its disregard of SSF.
In \cite{LSF_CSI_Power2025}, it is shown that the performance loss from large-timescale power allocation is slight due to channel hardening in massive MIMO systems. Nonetheless, only power allocation was optimized therein. Whether or not the net-SE loss from jointly optimizing pilot assignment and power allocation is also negligible remains unclear.

No numerical algorithms are available for these two problems in the literature. In fact, the DTS problem in \eqref{eq:subproblems} has never been addressed before. The STS problem in \eqref{eq:problem-LSF} is different from those in \cite{2024JPCPA-CF-Ren,2024JPCPA-CF-Khan,JPAPC2025}, where the pilot length is not optimized.

Both the DTS and STS problems involve continuous and integer variables, making numerical algorithms computationally intensive.
Hence, we resort to the learning-based approach, which finds the mapping from the environmental states to the optimal solution of a problem.

\subsection{Three Policies to be Learned}

From the DTS problem, we jointly learn a power allocation policy and a pilot assignment policy, which are respectively the mappings from the environmental states to the optimal power allocation matrix (denoted by ${\bf P}^{t}_{\sf DTS}$) and the optimal pilot assignment matrix (denoted by ${\bf X}_{\sf DTS}$). From the STS problem, we learn an STS policy, which is the mapping from the environmental states to the optimal solution of the problem in \eqref{eq:problem-LSF} (denoted by
$[{\bf X}_{\sf STS},{\bf P}_{\sf STS}]$). Next, we analyze the environmental states of the three policies.
\subsubsection{Power allocation policy}
From \eqref{eq:low-level}, \eqref{eq:net-SE}, and \eqref{eq:SINR}, we can see that the optimal power allocation ${\bf P}^{t}_{\sf DTS}$ depends on the equivalent channel matrix ${\bf G}^t$, where ${\bf G}^{t}=[({{\bf G}_{1}}^{t})^{\sf T},\cdots,\\({{\bf G}_{M}}^{t})^{\sf T}]\in{\mathbb C}^{K\times MK}$, and ${\bf G}_{m}^{t}=[{g}_{m_ik}^{t}]_{K\times K}\in{\mathbb C}^{K\times K}$.

${g}_{m_ik}^{t}=({{\sqrt{\beta_{mk}}}\bf h}_{mk}^{t})^{\rm H}{\bf v}_{m_i}^{t}$ can be computed using the estimated local channels and the LSF channel gains.
The equivalent channels from ${\sf AP}_m$ to its associated UEs are computed by using the estimated local channels from \eqref{eq:estimated-channel}. The equivalent channels from ${\sf AP}_m$ to its non-associated UEs are obtained by taking the expectation over the SSF channels, i.e.,
\begin{align}
&{\mathbb E}\{|({\sqrt{\beta_{mk}}{\bf h}}_{mk}^t)^{\rm H} {\bf v}_{m_i}^t|\}&\nonumber\\
=&\sqrt{\!\frac{1}{N}{\sf tr}({\mathbb E}\{(\sqrt{\beta_{mk}}{\bf h}_{mk}^t)(\sqrt{\beta_{mk}}{{\bf h}}_{mk}^t)^{\rm H}\!\}){\mathbb E}\{\!({\bf v}_{m_i}^t)^{\rm H}{\bf v}_{m_i}^t\!\}}\nonumber\\
=&\sqrt{\beta_{mk}}.
\end{align}
Hence, the estimated equivalent channel coefficient is
\begin{equation}\label{eq:estimated-equivalent-channel}
	{\hat g}_{m_ik}^t=
	\begin{cases}
		({\hat{\bf h}}_{mk}^t)^{\rm H} {\bf v}_{m_i}^t,&~\forall i,k\in{\mathcal A}_m\\
		\sqrt{\beta_{mk}}, &~\forall i\in{\mathcal A}_m,k\notin{\mathcal A}_m\\
		0,&~\forall i\notin{\mathcal A}_m
	\end{cases}.
\end{equation}

Then, the power allocation policy is
\begin{align}\label{eq:Power_policy}
{{\bf P}_{\sf DTS}^t}=f_{\sf DTS}^{\sf power}({\hat{\bf G}^t}).
\end{align}

\subsubsection{Pilot assignment policy}
The optimal pilot assignment only depends on the LSF channel gain matrix and the association matrix.
Hence, the pilot assignment policy is
\begin{align}\label{eq:Pilot_policy}
{\bf X}_{\sf DTS}=f_{\sf DTS}^{\sf pilot}({\boldsymbol\beta},{\bf A}).
\end{align}

\subsubsection{STS policy}
  From the problem in \eqref{eq:problem-LSF}, we know that both pilot assignment and power allocation depend on $\boldsymbol\beta$ and $\bf A$. Hence, the STS policy is
\begin{align}\label{eq:STS_policy}
[{\bf X}_{\sf STS},{\bf P}_{\sf STS}]=f_{\sf STS}({\boldsymbol\beta},{\bf A}),
\end{align}
which contains two policies, denoted by ${\bf X}_{\sf STS}=f_{\sf STS}^{\sf pilot}({\boldsymbol\beta},{\bf A})$ and ${\bf P}_{\sf STS}=f_{\sf STS}^{\sf power}({\boldsymbol\beta},{\bf A})$, respectively.

\vspace{-1mm}\subsection{Permutation Properties}\label{sec:PE}

Using the method in \cite{LSJ2023MGNN}, we can analyze permutation properties of the DTS and STS policies by first identifying the sets involved in the corresponding policies.

\subsubsection{Power allocation policy}\label{sec:policy_two}
From the problem in \eqref{eq:low-level}, we can identify two sets: the UE- and AN-sets.
The UE-set is a normal set whose elements can be reordered arbitrarily (i.e., reordering these elements does not change the optimization problem).
The AN-set is a nested set consisting of $M$ subsets due to the power constraint at each AP, i.e.,
\begin{align}
\{\{{\sf ANs~at}~{\sf AP}_1\},\cdots,\{{\sf ANs~ at}~ {\sf AP}_M\}\}.
\end{align}
The subsets can be reordered arbitrarily as a whole.
In each subset, the ANs can be reordered together with their served UEs.
The permutation matrix for reordering the ANs is \cite{LSJ2023MGNN}
\begin{align}
{\bf\Omega}_{\sf AN}=({\bf I}_M\otimes {\bf\Pi}_{\sf UE})\cdot({\bf \Pi}_{\sf AP}\otimes {\bf I}_{K})\in\{0,1\}^{MK\times MK},
\end{align}
where ${\bf \Pi}_{\sf UE}\in\{0,1\}^{K\times K}$ and ${\bf \Pi}_{\sf AP}\in\{0,1\}^{M\times M}$ respectively represent the reordering of the UEs and APs (i.e., subsets, since a subset corresponds to an AP). For example, when $M=2$, $K=2$, and ${\bf\Pi}_{\sf AP}={\bf\Pi}_{\sf UE}=\left( \begin{smallmatrix}
	0 & 1 \\
	1 & 0
\end{smallmatrix} \right)$ (which represents swapping ${\sf AP}_1$ and ${\sf AP}_2$, and ${\sf UE}_1$ and ${\sf UE}_2$), the ANs are reordered from $\{\{{\sf AN}_{1_1},{\sf AN}_{1_2}\},\{{\sf AN}_{2_1},{\sf AN}_{2_2}\}\}$ to $\{\{{\sf AN}_{2_2},{\sf AN}_{2_1}\},\{{\sf AN}_{1_2},{\sf AN}_{1_1}\}\}$.

By vectorizing the  power allocation matrix ${\bf P}_{\sf DTS}^t$ as ${\vec{\bf P}}_{\sf DTS}^t\!\!=\!\![p_{{\sf DTS},11}^t,\cdots\!,p_{{\sf DTS},1K}^t,\cdots\!,p_{{\sf DTS},M1}^t,\cdots\!,p_{{\sf DTS},MK}^t]\in\!\!{\mathbb R}^{MK}$, we can show that
${\vec{\bf P}}_{\sf DTS}^t$ and ${\hat{\bf G}^t}$ are respectively permuted as ${\vec{\bf P}}_{\sf DTS}^t{\bf\Omega}_{\sf AN}$ and ${\bf\Pi}_{\sf UE}^{\sf T}{\hat{\bf G}^t}{\bf\Omega}_{\sf AN}$ when UEs and ANs are reordered separately. Hence, the power allocation policy in \eqref{eq:Power_policy} satisfies
\begin{equation}\label{eq:PE-DTS-power}
{\vec{\bf P}}_{\sf DTS}^t{\bf\Omega}_{\sf AN}=f_{\sf DTS}^{\sf power}({\bf\Pi}_{\sf UE}^{\sf T}{\hat{\bf G}^t}{\bf\Omega}_{\sf AN}).
\end{equation}

\subsubsection{Pilot assignment policy}
From the problem in \eqref{eq:high-level}, we can identify three normal sets: the PS-, AP-, and UE-sets. These sets are independent of each other, i.e., the elements in each set can be reordered no matter how the elements in the other two sets are reordered.

When the PSs, APs, and UEs are reordered separately,
${\bf X}_{\sf DTS}$, ${\boldsymbol\beta}$, and ${\bf A}$ are permuted as ${\bf \Pi}_{\sf PS}^{\sf T}{\bf X}_{\sf DTS}{\bf \Pi}_{\sf UE}$, ${\bf \Pi}_{\sf AP}^{\rm T}{\boldsymbol\beta}{\bf \Pi}_{\sf UE}$, and ${\bf \Pi}_{\sf AP}^{\rm T}{\bf A}{\bf \Pi}_{\sf UE}$, respectively, where ${\bf\Pi}_{\sf PS}\in\{0,1\}^{K\times K}$. Hence, the pilot assignment policy in \eqref{eq:Pilot_policy} satisfies
\begin{equation}\label{eq:PE-DTS-pilot}
{\bf \Pi}_{\sf PS}^{\sf T}{\bf X}_{\sf DTS}{\bf \Pi}_{\sf UE}=f_{\sf DTS}^{\sf pilot}({\bf \Pi}_{\sf AP}^{\rm T}{\boldsymbol\beta}{\bf \Pi}_{\sf UE},{\bf \Pi}_{\sf AP}^{\rm T}{\bf A}{\bf \Pi}_{\sf UE}).
\end{equation}


\subsubsection{STS policy}
Similarly, we can identify three normal sets from the problem in \eqref{eq:problem-LSF}: the PS-, AP-, and UE-sets, and these sets are independent of each other, the same as the pilot assignment policy.
Thereby, the STS policy in \eqref{eq:STS_policy} satisfies
\begin{equation}\label{eq:PE-STS}
\begin{aligned}
[{\bf\Pi}_{\sf PS}^{\sf T}{\bf X}_{\sf STS}{\bf\Pi}_{\sf UE},{\bf\Pi}_{\sf AP}^{\sf T}{\bf P}_{\sf STS}{\bf\Pi}_{\sf UE}]=\\f_{\sf STS}({\bf\Pi}_{\sf AP}^{\sf T}{\boldsymbol\beta}{\bf\Pi}_{\sf UE},{\bf\Pi}_{\sf AP}^{\sf T}{\bf A}{\bf\Pi}_{\sf UE}).
\end{aligned}
\end{equation}

\vspace{-1mm}\subsection{One-to-many Mappings and Their Impact}\label{sec:Mapping}
The pilot assignment policies obtained from both DTS and STS problems can be expressed in a unified form as
\begin{equation}
{\bf X}^*=f^{\sf pilot}({\boldsymbol\beta},{\bf A}),
\end{equation}
where ${\bf X}^*$ is an optimal pilot assignment matrix.

Next, we show that the pilot assignment policies from both problems are one-to-many mappings.

\begin{Proposition}\label{Prop:proposition 1}
${\bf \Pi}_{\sf PS}^{\sf T}{\bf X}^*$ is also the optimal pilot assignment matrix, i.e.,
\begin{equation}\label{eq:proposition 1-1}
{\bf \Pi}_{\sf PS}^{\sf T}{\bf X}^*\!=\!f^{\sf pilot}({\boldsymbol\beta},{\bf A}).
\end{equation}
When $\beta_{mi}=\beta_{mj}$ and $a_{mi}=a_{mj}$, $i\neq j$, $\forall m$, ${\bf X}^*{\bf \Pi}^{\sf S}_{\sf UE}$ is also the optimal pilot assignment matrix, i.e.,
\begin{equation}\label{eq:proposition 1-2} {\bf X}^*{\bf\Pi}^{\sf S}_{\sf UE}=f^{\sf pilot}({\boldsymbol\beta}, {\bf A}),
\end{equation}
where ${\bf\Pi}^{\sf S}_{\sf UE}\!\!=\!\!{\bf I}_K$ or  ${\bf \Pi}^{\sf S}_{\sf UE}\!\!=\!\![{\bf e}_{11},\cdots,{\bf e}_{ij},\cdots,{\bf e}_{ji},\cdots,{\bf e}_{KK}]$ are two specific values of ${\bf\Pi}_{\sf UE}$, and ${\bf e}_{ij}$ is the $i$th column vector of ${\bf \Pi}^{\sf S}_{\sf UE}$ and is a standard basis vector with one at the $j$th entry and zeros elsewhere.
\end{Proposition}
\begin{proof}
From \eqref{eq:PE-DTS-pilot} and \eqref{eq:PE-STS}, we have ${\bf\Pi}_{\sf PS}^{\sf T}{\bf X}^*{\bf\Pi}_{\sf UE}=f^{\sf pilot}({\bf\Pi}_{\sf AP}^{\sf T}{\boldsymbol\beta}{\bf\Pi}_{\sf UE},{\bf\Pi}_{\sf AP}^{\sf T}{\bf A}{\bf\Pi}_{\sf UE})$.

When ${\bf\Pi}_{\sf AP}={\bf I}_{M}$ and ${\bf\Pi}_{\sf UE}={\bf I}_{K}$, we obtain \eqref{eq:proposition 1-1}.

When ${\bf\Pi}_{\sf PS}={\bf I}_{{K}}$, ${\bf\Pi}_{\sf AP}={\bf I}_{M}$, and ${\bf\Pi}_{\sf UE}={\bf \Pi}^{\sf S}_{\sf UE}$, we obtain ${\bf X}^*{\bf \Pi}^{\sf S}_{\sf UE}=f^{\sf pilot}({\boldsymbol\beta{\bf \Pi}^{\sf S}_{\sf UE}},{\bf A}{\bf \Pi}^{\sf S}_{\sf UE})$.
If ${\bf \Pi}^{\sf S}_{\sf UE}={\bf I}_K$, \eqref{eq:proposition 1-2} holds.
If ${\bf \Pi}^{\sf S}_{\sf UE}=[{\bf e}_{11},\cdots,{\bf e}_{ij},\cdots,{\bf e}_{ji},\cdots,{\bf e}_{KK}]$, ${\boldsymbol\beta}{\bf \Pi}^{\sf S}_{\sf UE}$ and ${\bf A}{\bf \Pi}^{\sf S}_{\sf UE}$ are the matrices obtained after swapping the $i$th and $j$th columns of ${\boldsymbol\beta}$ and ${\bf A}$, respectively.
Under the conditions that $\beta_{mi}=\beta_{mj}$ and $a_{mi}=a_{mj}$, $\forall m$, the $i$th and the $j$th columns of $\boldsymbol\beta$ and ${\bf A}$ are the same, hence ${\boldsymbol\beta}={\boldsymbol\beta}{\bf \Pi}^{\sf S}_{\sf UE}$ and ${\bf A}={\bf A}{\bf \Pi}^{\sf S}_{\sf UE}$. Then, \eqref{eq:proposition 1-2} holds.
When ${\bf \Pi}^{\sf S}_{\sf UE}\neq{\bf I}_K$ and ${\bf \Pi}^{\sf S}_{\sf UE}\neq[{\bf e}_{11},\cdots,{\bf e}_{ij},\cdots,{\bf e}_{ji},\cdots,{\bf e}_{KK}]$, since ${\boldsymbol\beta}\neq{\boldsymbol\beta}{\bf \Pi}^{\sf S}_{\sf UE}$ and ${\bf A}\neq{\bf A}{\bf \Pi}^{\sf S}_{\sf UE}$, we have ${\bf X}^*{\bf \Pi}^{\sf S}_{\sf UE}=f^{\sf pilot}({\boldsymbol\beta}{\bf \Pi}^{\sf S}_{\sf UE},{\bf A}{\bf \Pi}^{\sf S}_{\sf UE})\neq f^{\sf pilot}({\boldsymbol\beta},{\bf A})$. Hence, only for the two specific values of ${\bf\Pi}_{\sf UE}$, \eqref{eq:proposition 1-2} holds.
\end{proof}

There exists $K!$ possible values of ${\bf \Pi}_{\sf PS}$. Hence, \eqref{eq:proposition 1-1} indicates that $f^{\sf pilot}(\cdot)$ is a one-to-many  mapping. Since there are two possible values of ${\bf \Pi}^{\sf S}_{\sf UE}$, \eqref{eq:proposition 1-2} also shows that $f^{\sf pilot}(\cdot)$ is a one-to-many mapping.

To help understand, we use the same example as in Fig. \ref{fig:pilot_assignment_matrix} with $K=3$ to illustrate the multiple optimal pilot assignment matrices for the same ${\boldsymbol\beta}$ and ${\bf A}$.

As shown in Fig. \ref{fig:permutation}(a), after permuting ${\bf X}^*$ along the PS dimension (i.e., reordering the PSs) by ${\bf \Pi}_{\sf PS}^{\sf T}{\bf X}^*$, the UEs that were originally assigned ${\sf PS}_1$ (i.e., ${\sf UE}_1$ and ${\sf UE}_3$) are now assigned ${\sf PS}_2$, and the UEs that were assigned ${\sf PS}_2$ (i.e., ${\sf UE}_2$) are assigned ${\sf PS}_1$.

When ${\bf \Pi}^{\sf S}_{\sf UE}=[{\bf e}_{11},\cdots,{\bf e}_{ij},\cdots,{\bf e}_{ji},\cdots,{\bf e}_{KK}]$, ${\bf X}^*{\bf \Pi}^{\sf S}_{\sf UE}$ is a matrix obtained by permuting ${\bf X}^*$ along the UE dimension (i.e., swapping the PSs assigned to ${\sf UE}_i$ and ${\sf UE}_j$). As shown in Fig. \ref{fig:permutation}(b), if $i=2$ and $j=3$, after the permutation, ${\sf UE}_2$ and ${\sf UE}_3$, which were originally assigned ${\sf PS}_2$ and ${\sf PS}_1$, respectively, are now respectively assigned ${\sf PS}_1$ and ${\sf PS}_2$.

\vspace{-2mm}
\begin{figure}[!htp]
	\centering
	\includegraphics[scale=0.6]{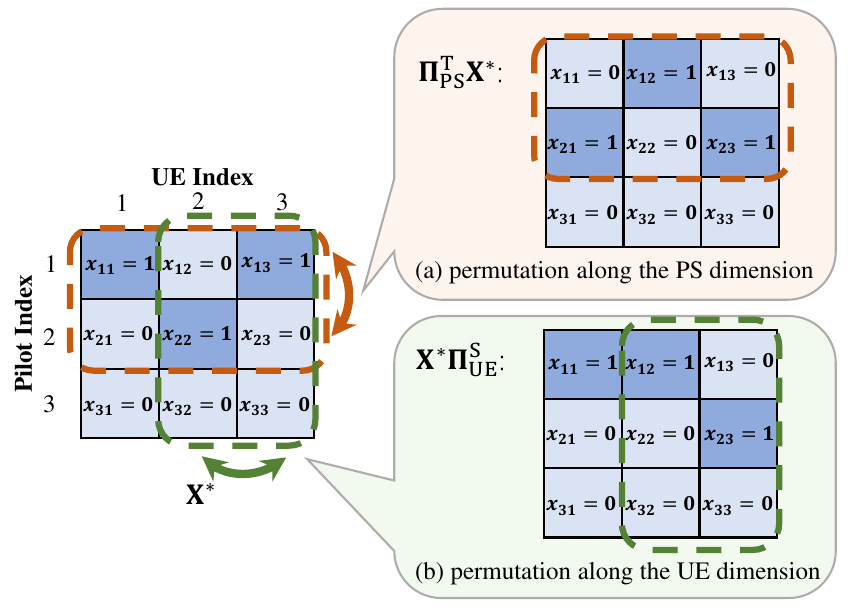}
	\vspace{-2mm}
	\caption{Three optimal pilot
assignment matrices for same ${\boldsymbol\beta}$ and ${\bf A}$.} \label{fig:permutation}
	\vspace{-2mm}
\end{figure}

In what follows, we analyze the challenges of learning the one-to-many mappings using GNNs. As a kind of DNN, GNNs learn functions, which are one-to-one mappings or many-to-one mappings. If the pilot assignment policy learned by a GNN is $\hat{\bf X}=\hat{f}^{\sf pilot}({\boldsymbol\beta},{\bf A})$ and the GNN satisfies the same permutation property as the pilot assignment policy (hence $\hat{f}^{\sf pilot}(\cdot)$ also satisfies \eqref{eq:proposition 1-1} and \eqref{eq:proposition 1-2}), then the following proposition indicates that $\hat{\bf X}$ is far from optimal.

\begin{Proposition}\label{Prop:proposition 2}
If $\hat{\bf X}=\hat{f}^{\sf pilot}({\boldsymbol\beta},{\bf A})$ satisfies \eqref{eq:proposition 1-1} and \eqref{eq:proposition 1-2}, then the output pilot assignment of the GNN is
\begin{subequations}
\begin{align}\label{eq:proposition 2-1}
&{\hat x}_{1k}=\cdots={\hat x}_{Kk}, ~\forall k,\\\label{eq:proposition 2-2}
&\hat{\bf x}_{i}=\hat{\bf x}_{j},~{\text{if}}~\beta_{mi}=
\beta_{mj},a_{mi}=a_{mj}, \forall m
\end{align}
\end{subequations}
\end{Proposition}
\begin{proof}
If ${\hat f}^{\sf pilot}(\cdot)$ satisfies \eqref{eq:proposition 1-1}, then ${\bf\Pi}_{\sf PS}^{\sf T}{\bf\hat X} = {\hat f}^{\sf pilot}({\boldsymbol\beta},{\bf A})$, $\forall {\bf\Pi}_{\sf PS}$. Since ${\hat f}^{\sf pilot}(\cdot)$ is a function (i.e., one pair of $({\boldsymbol\beta},{\bf A})$ can only map into a single output),
we have ${\bf\Pi}_{\sf PS}^{\sf T}{\bf\hat X}={\bf\hat X}$, from which we know that the elements in each column of ${\bf\hat X}$ are
identical, i.e., ${\hat x}_{1k} =\cdots= {\hat x}_{Kk}, \forall k$.

If ${\hat f}^{\sf pilot}(\cdot)$ satisfies \eqref{eq:proposition 1-2}, then ${\bf\hat X}{\bf \Pi}^{\sf S}_{\sf UE}={\hat f}^{\sf pilot}({\boldsymbol\beta},{\bf A})$ when $\beta_{mi}=\beta_{mj}$ and $a_{mi}=a_{mj}$, $\forall m$. Again, since ${\hat f}^{\sf pilot}(\cdot)$ is a function, we have ${\bf\hat X}{\bf \Pi}^{\sf S}_{\sf UE} = {\bf\hat X}$. When ${\bf \Pi}^{\sf S}_{\sf UE}\!\!=\!\![{\bf e}_{11},\cdots,{\bf e}_{ij},\cdots,{\bf e}_{ji},\cdots,{\bf e}_{KK}]$, ${\bf\hat X}{\bf \Pi}^{\sf S}_{\sf UE}$ is a matrix obtained after swapping the $i$th and $j$th columns of ${\bf\hat X}$. Hence, the $i$th and $j$th columns of ${\bf\hat X}$ are identical, i.e.,  ${\bf\hat x}_i = {\bf\hat x}_j$.
\end{proof}


\eqref{eq:proposition 2-1} means that ${\hat x}_{gk}=1$, $\forall g, k$, considering the constraint in \eqref{eq:cons-1}. Then, all PSs are assigned to each UE. However, this assignment is infeasible because each UE can be assigned only one PS, as constrained by \eqref{eq:cons-1}.

\eqref{eq:proposition 2-2} means that one PS will be assigned to multiple UEs associated with the same AP if they have identical LSF gains. However, different PSs will be assigned to such UEs by the optimal pilot assignment to avoid severe mutual
pilot contamination. In practice, it is unlikely for two UEs to have exactly the same LSF channel gain. Nonetheless, since a GNN produces continuous functions, $\hat{\bf x}_{i}\approx\hat{\bf x}_{j}$ if $\beta_{mi}\approx\beta_{mj}$ and $a_{mi}=a_{mj}$ \cite{2025SFSA-LSJ}. The possibility of $\hat{\bf x}_{i}=\hat{\bf x}_{j}$ increases given that the pilot assignment yielded by the GNN needs to be discretized to 0 or 1.

A similar issue to \eqref{eq:proposition 2-2} was found in the GNN-based user scheduling in \cite{2025SFSA-LSJ}, where two UEs are either both scheduled or both unscheduled if their channels are identical. Since the probability of similar SSF is much lower than that of similar LSF gains, this issue has a large impact on the scheduling performance only when the UEs are densely located. Yet the issue in \eqref{eq:proposition 2-1} does not occur in \cite{2025SFSA-LSJ}, which will incur substantial performance degradation if the GNN is not well-designed for learning pilot assignment.

\section{Designing GNNs}\label{sec:GNN}
In this section, we design GNNs to optimize the power allocation and pilot assignment. To reduce the training complexity and enable size generalizability, we first recap a method to design a GNN for satisfying the permutation property of a policy (called the desired permutation property). Then, we design three GNNs for learning the power allocation, pilot assignment, and STS policies, which are referred to as DTS-Power-GNN, DTS-Pilot-GNN, and STS-GNN, respectively. To address the challenge faced by the GNNs for pilot assignment, we resort to feature enhancement. To improve learning performance, we design an attention mechanism for reflecting potential pilot contamination. Finally, we introduce the training and testing procedures.

\vspace{-2mm} \subsection{Designing a GNN with Desired Permutation Property}\label{sec:KeyPoints}
A graph consists of vertices and edges, each of which may be associated with features or actions. A GNN learns a policy by mapping features to actions over a graph through multiple update layers. To ensure that the GNN satisfies the desired permutation property, the graph needs to be constructed properly, and the update equation in each layer should be carefully designed  \cite{LSJ2023MGNN,2025ZJYDesignGNN}.

Since the permutation property of a policy is induced by sets and the functions learned by a GNN are induced by permutable vertices, the elements in each set are defined as one type of vertices \cite{LSJ2023MGNN}. The edges, along with features and actions of vertices or edges, can be determined by examining the input and output dimensions of the policy. To reflect the permutability across related sets, the permutable edges are defined as the same type of edge \cite{2025ZJYDesignGNN}.

After the graph is constructed, either a vertex-GNN or an edge-GNN can be used, which either updates the representations of vertices or updates the representations of edges. Since the actions of the graphs to be constructed for the three policies are defined on edges, we consider edge-GNNs.
According to the proof in \cite{2025ZJYDesignGNN}, an edge-GNN learning over a graph constructed properly will exhibit the desired permutation properties if the following three conditions are satisfied.
\begin{enumerate}[label=(\roman*), leftmargin=*,labelsep=0.1em, align=left, nosep]
    \item  The processing functions for the edges to aggregate information are identical when the edges and their neighbors are of the same types.
    \item The combination functions are identical for the edges of the same type.
    \item The pooling functions satisfy the commutative law and are identical for the edges of the same type.
\end{enumerate}

In the following subsections, we construct the graphs and design the GNNs to satisfy the permutation properties in \eqref{eq:PE-DTS-power}, \eqref{eq:PE-DTS-pilot}, and \eqref{eq:PE-STS} using the aforementioned method.

Various processing, pooling, and combination functions can be selected. We adopt mean pooling and a combination function formed by cascading a linear function and an activation function. Since the pooling function is identical for every edge, condition (iii) is satisfied. We adopt linear processing, unless otherwise specified.

\vspace{-2mm} \subsection{Designing DTS-Power-GNN}\label{sec:DTS-GNN}
This GNN learns the power allocation
policy ${\vec{\bf P}_{\sf DTS}^t}=f_{\sf DTS}^{\sf power}({\hat{\bf G}^t})$, where ${\hat{\bf G}}^{t}=[({\hat{\bf G}_{1}}^{t})^{\sf T},\cdots,({\hat{\bf G}_{M}}^{t})^{\sf T}]\in{\mathbb C}^{K\times MK}$, and ${\hat{\bf G}}_{m}^{t}=[{\hat g}_{m_ik}^{t}]_{K\times K}$. Its output is the power allocation vector $\vec{{\hat { {\bf P}}}}_{\sf DTS}^t\!=\![{\hat p}_{{\sf DTS},11}^t,\cdots\!,{\hat p}_{{\sf DTS},1K}^t,\cdots\!,{\hat p}_{{\sf DTS},M1}^t,\cdots\!,{\hat p}_{{\sf DTS},MK}^t]\\\in{\mathbb R}^{ MK}$ in a subframe.

\subsubsection{Graph construction}
Since the policy is defined on two sets (i.e., the AN- and UE-sets), we construct a graph with two types of vertices: AN-vertices and UE-vertices, as illustrated in Fig. \ref{fig:graph_CSI}. Specifically, the equivalent antenna serving a UE is defined as an AN-vertex, and each UE is defined as a UE-vertex. No vertex has features or actions.

\vspace{-2mm}
\begin{figure}[!htp]
	\centering
	\includegraphics[scale=0.35]{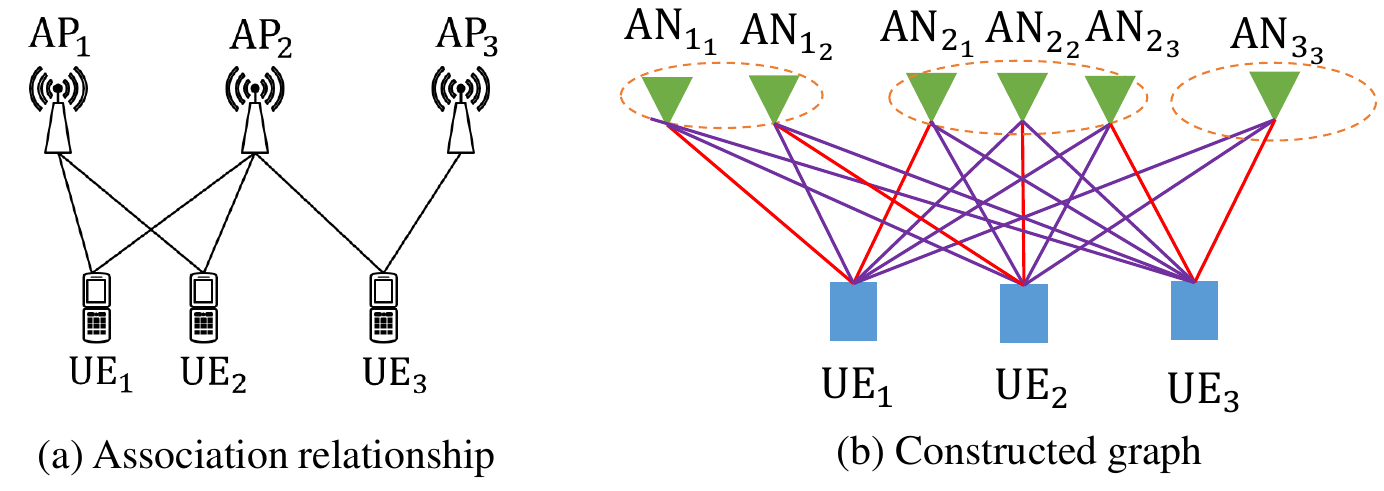}
    \vspace{-1mm}
	\caption{User association and graph for the DTS-Power-GNN. The red and purple lines represent SIG and INF edges, respectively.
} \label{fig:graph_CSI}
    \vspace{-4mm}
\end{figure}

The environment state ${\bf\hat G}^t$ has AN and UE dimensions. Hence, an edge exists between each AN-vertex (say ${\sf AN}_{m_i}$-vertex) and each UE-vertex (say ${\sf UE}_k$-vertex), denoted as edge $({\sf AN}_{m_i}, {\sf UE}_k)$, whose feature is ${{\hat g}_{m_ik}^t}$.

To satisfy the permutation property in \eqref{eq:PE-DTS-power}, the ordering of AN-vertices should be tied to the ordering of UE-vertices. For example, if ${\sf UE}_1$ and ${\sf UE}_2$ in Fig. \ref{fig:graph_CSI}(b) are swapped, the ANs serving ${\sf UE}_1$ should be swapped with those serving ${\sf UE}_2$ (e.g., ${\sf AN}_{1_1}$ is swapped with ${\sf AN}_{1_2}$). Then, edge $({\sf AN}_{1_1},{\sf UE}_1)$ and edge $({\sf AN}_{1_2},{\sf UE}_2)$ are swapped, and edge $({\sf AN}_{1_2},{\sf UE}_1)$ and edge $({\sf AN}_{1_1},{\sf UE}_2)$ are swapped. This indicates that edges $({\sf AN}_{m_k},{\sf UE}_k)$, $\forall m, k\in{{\mathcal A}_m}$ can be permuted among themselves, and the remaining edges can be permuted among themselves, but the two groups of edges cannot be permuted. Therefore, two types of edges are defined. Edges $(\text{AN}_{m_k}, \text{UE}_k)$ for every $k \in \mathcal A_m$ are called signal (SIG) edges, while all other edges are called interference (INF) edges.

%
Since $\vec{\bf P}_{\sf DTS}^t\in{\mathbb R}^{ MK}$ has AN dimension and ${\sf AN}_{m_k}$ only allocates power to ${\sf UE}_k$, the actions are defined on $\rm SIG$ edges, where ${\hat p}_{{\sf DTS},mk}^t$ denotes the action on edge $({\sf AN}_{m_k},{\sf UE}_k)$.

\subsubsection{Updating equations design}
Since the features and actions of the constructed graphs are defined on edges, we adopt edge-GNNs, which achieve performance comparable to vertex-GNNs with low training and inference time \cite{2024EdgeGNN-Peng}.
In each layer of an edge-GNN, the hidden representation of an edge is updated by first aggregating information from neighboring edges through processing and pooling functions, and then combining its own information with the aggregated information via a combination function.

To meet conditions (i) and (ii), different weight matrices are used when updating representations of $\rm SIG$ and $\rm INF$ edges.

{\textbf{Updating representations of $\rm SIG$ edges:}} The neighboring edges of $\rm SIG$ edge $({\sf AN}_{m_k},{\sf UE}_k)$ include all $\rm INF$ edges connected to ${\sf AN}_{m_k}$-vertex and all $\rm INF$ and $\rm SIG$ edges connected to ${\sf UE}_{k}$-vertex, as illustrated in Fig. \ref{fig:graph_CSI}(b).

To satisfy condition (i), three processing functions are used, parameterized by different weight matrices. The first is for neighboring $\rm INF$ edges connecting to an AN-vertex, the second is for neighboring $\rm INF$ edges connecting to a UE-vertex, and the third is for neighboring $\rm SIG$ edges. To satisfy condition (ii), all $\rm SIG$ edges use the same weight matrix for combination. Consequently, the hidden representation vector of $\rm SIG$ edge $({\sf AN}_{m_k},{\sf UE}_k)$ is updated  in the $l$th layer as
\begin{equation}\label{eq:update-AN-UE-edges}
	\begin{aligned}
		&\!\!\!\!\!{\bf d}_{{\sf AN}_{m_k}\!,{\sf UE}_k}^{(l)}\!\!\!\!=\!\phi({\bf Q}_{{\sf AN},1}{\bf d}_{{\sf AN}_{m_k}\!,{\sf UE}_k}^{(l-1)}\!+\!\frac{1}{K}\sum_{\substack{i=1\\i\neq k}}^K\! {\bf U}_{{\sf AN},1}{\bf d}_{{\sf AN}_{m_k}\!,{\sf UE}_i}^{(l-1)}\\
		&\!\!\!\!\!+\!\!\frac{1}{M|{\mathcal A}_j|}\!\!\!\sum_{j=1}^M\!\sum_{\substack{i\in{\mathcal A}_j\\i\neq k}}\!\!{\bf U}_{{\sf AN},2}{\bf d}_{{\sf AN}_{j_i},{\sf UE}_k}^{(l-1)}\!\!\!+\!\frac{1}{|{\mathcal S}_k|}\!\!\sum_{\substack{j\in{\mathcal S}_k\\j\neq m}}\!\!{\bf U}_{{\sf AN},3}{\bf d}_{{\sf AN}_{j_k},{\sf UE}_k}^{(l-1)}\!),
	\end{aligned}
\end{equation}
where $\phi(\cdot)$ is an activation function, ${\bf Q}_{{\sf AN},1}$ is the trainable weight matrix in the combination function, and ${\bf U}_{{\sf AN},1}\sim{\bf U}_{{\sf AN},3}$ are the trainable weight matrices in different processing functions. We omit the superscript $(l)$ of the activation function and weight matrices for notational simplification, and this omission applies in the sequel.

{\textbf{Updating representations of $\rm INF$ edges:}} The neighboring edges of $\rm INF$ edge $({\sf AN}_{m_i},{\sf UE}_k)$, $i\neq k$ include both $\rm INF$ and $\rm SIG$ edges connected by ${\sf AN}_{m_i}$-vertex or ${\sf UE}_{k}$-vertex, as illustrated in Fig. \ref{fig:graph_CSI}(b).

To meet condition (i), four processing functions are used, parameterized by different weight matrices. The first and second are for neighboring $\rm INF$ edges connecting to an  AN-vertex or a UE-vertex, respectively. The third and fourth are for neighboring $\rm SIG$ edges connecting to an AN-vertex or a UE-vertex, respectively.

To meet condition (ii), all $\rm INF$ edges use the same weight matrix for combination. The update equations for the hidden representations of $\rm INF$ edges are similar to \eqref{eq:update-AN-UE-edges}, with differences only in the sets of neighboring edges and the corresponding weight matrices. For conciseness, we no longer provide the update equations.

The power allocation vector yielded by the DTS-Power-GNN is the edge representations in the $L$th layer. In order to satisfy the constraint in \eqref{eq:cons-2}, we adopt the output-layer activation function designed in \cite{DataRateModel_GJ2023}.


\subsection{Designing DTS-Pilot-GNN}
This GNN learns the pilot assignment policy ${\bf X}_{\sf DTS}=f_{\sf DTS}^{\sf pilot}({\boldsymbol\beta},{\bf A})$, where ${\boldsymbol\beta}\in{\mathbb R}^{M\times K}$ and ${\bf A}\in{\mathbb R}^{M\times K}$. Its output is the pilot assignment matrix $\hat{\bf X}_{\sf DTS}=[{\hat x}_{{\sf DTS},gk}]_{K\times K}$.
\subsubsection{Graph construction} \label{CSI-pilot-graph}
Since the policy is defined on three sets (i.e., the AP-, UE-, and PS-sets), we construct a graph with three types of vertices: AP-, UE-, and PS-vertices, as illustrated in Fig. \ref{fig:graph_LSF}. No vertices have features or actions.

\vspace{-2mm}
\begin{figure}[!htp]
	\centering
	\includegraphics[scale=0.43]{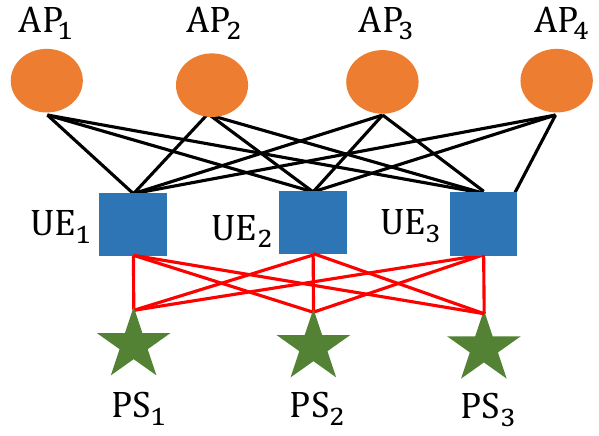}
	\vspace{-1mm}
	\caption{Graph constructed for the DTS-Pilot-GNN, where black and red lines represent AP-UE and PS-UE edges, respectively.
	} 	\label{fig:graph_LSF}
    \vspace{-2mm}
\end{figure}

The environmental states ${\boldsymbol\beta}$ and ${\bf A}$ have AP and UE dimensions. Hence, an edge exists between each AP-vertex and each UE-vertex (called AP-UE edge). Since ${\bf X}_{\sf DTS}\in\{0,1\}^{K\times K}$ has PS and UE dimensions, an edge exists between each PS-vertex and each UE-vertex  (called PS-UE edge).
To satisfy the permutation property in \eqref{eq:PE-DTS-pilot}, the AP-, UE-, and PS-vertices should be reordered independently. Thus, the AP-UE edges belong to a type, and the PS-UE edges belong to another type.

The AP-UE edges have features (i.e., ${\boldsymbol\beta}$ and ${\bf A}$) but do not have action. The feature vector of the edge between ${\sf AP}_m$-vertex and ${\sf UE}_k$-vertex is  $[{\beta}_{mk},{a}_{mk}]^{\sf T}$. The PS-UE edges have no feature but have action (i.e., ${\bf \hat X}_{\sf DTS}$). The action on edge $({\sf PS}_g,{\sf UE}_k)$ is ${\hat x}_{{\sf DTS},gk}$.
\subsubsection{Feature enhancement}\label{sec:FE}
Proposition \ref{Prop:proposition 2} means that the GNN with input ${\boldsymbol\beta},{\bf A}$ for pilot assignment cannot distinguish the edges whose actions should be different. Since ${\hat x}_{1k},\cdots,{\hat x}_{{K}k}$ are actions on edges $({\sf PS}_g,{\sf UE}_k)$, $\forall g$, \eqref{eq:proposition 2-1} indicates that the GNN cannot distinguish the PS-UE edges connecting to the same UE-vertex (e.g., edge $({\sf PS}_1,{\sf UE}_1)$, edge $({\sf PS}_2,{\sf UE}_1)$, and edge $({\sf PS}_3,{\sf UE}_1)$ in Fig. \ref{fig:graph_LSF}). Since ${\hat{\bf x}}_i$ and ${\hat{\bf x}}_j$ comprises actions on edges $({\sf PS}_g,{\sf UE}_i)$ and edges $({\sf PS}_g,{\sf UE}_j)$, $\forall g$, respectively, \eqref{eq:proposition 2-2} indicates that the GNN cannot distinguish the edges connecting a PS-vertex to ${\sf UE}_i$-vertex and ${\sf UE}_j$-vertex. For the example  in Fig. \ref{fig:graph_LSF}, if $i=1$ and $j=2$, edge $({\sf PS}_1,{\sf UE}_1)$ and edge $({\sf PS}_1,{\sf UE}_2)$ are undistinguishable, edge $({\sf PS}_2,{\sf UE}_1)$ and edge $({\sf PS}_2,{\sf UE}_2)$ are undistinguishable, and edge $({\sf PS}_3,{\sf UE}_1)$ and edge $({\sf PS}_3,{\sf UE}_2)$ are undistinguishable.

As analyzed in Section \ref{sec:Mapping}, the GNN without the distinguishing ability cannot learn the pilot assignment policy well. To address this issue, we consider a simple yet effective technique that augments the inputs of GNNs with random features \cite{RandomEquivalent,ExpressiveSurvey2025}. Specifically, we introduce artificial features into PS-UE edges by defining a random matrix $\boldsymbol{\Lambda}=[\lambda_{gk}]_{K\times K}$, where edge $({\sf PS}_g,{\sf UE}_k)$ has an artificial feature ${\lambda}_{gk}$.
Then, the designed GNN learns ${{\bf X}_{\sf DTS}}={f}_{\sf DTS}^{\sf pilot^+}({\boldsymbol\beta},{\bf A},\boldsymbol{\Lambda})$, which is a one-to-one mapping, rather than the one-to-more mapping ${\bf X}_{\sf DTS}=f_{\sf DTS}^{\sf pilot}({\boldsymbol\beta},{\bf A})$.

\subsubsection{Updating equations design}\label{sec:updating-equation}
Since the features and actions are defined on edges, we adopt edge-GNNs. To meet conditions (i) and (ii), different weight matrices are
used when updating representations of AP-UE and PS-UE edges.

{\textbf {Updating representations of AP-UE edges:}}
The neighboring edges of an AP-UE edge include both AP-UE and PS-UE edges. For example, as shown in Fig. \ref{fig:graph_LSF}, the neighboring edges of edge $({\sf AP}_1,{\sf UE}_1)$ are edges $({\sf AP}_m,{\sf UE}_1)$ for $m=2,3,4$, edges $({\sf AP}_1,{\sf UE}_k)$ for $k=2,3$, and edges $({\sf PS}_g,{\sf UE}_1)$ for $g=1,2,3$.
Since the features of PS-UE edges are artificially introduced for enhancing the distinguishing ability, the information from these edges is not useful for updating the representations of AP-UE edges. Thereby, only the information from neighboring AP-UE edges is aggregated.

To satisfy condition (i), two processing functions are used, parameterized by different weight matrices: one is for neighboring AP-UE edges connecting to a UE-vertex, and the other is for those connecting to an AP-vertex. To satisfy condition (ii), all AP-UE edges use the same weight matrix for combination. Then, the hidden representation of edge $({\sf AP}_m,{\sf UE}_k)$ is updated  in the $l$th layer as
\begin{equation} \label{eq:update-AP-UE-edges}
	\begin{aligned}
	\!\!\!\!\!{\bf d}_{{\sf AP}_m,{\sf UE}_k}^{(l)}\!\!\!=&\phi\!\left({\bf Q}_{{\sf AP},1}{\bf d}_{{\sf AP}_m,{\sf UE}_k}^{(l-1)}\!\!+\!\!\frac{1}{M}\!\!\!\sum\nolimits_{\substack{i=1\\i\neq m}}^{{M}}\!\!{\bf U}_{{\sf AP},1}{\bf d}_{{\sf AP}_i,{\sf UE}_k}^{(l-1)}\right.\\& +\left.\frac{1}{K}\sum\nolimits_{\substack{j=1,j\neq k}}^{{K}}{\bf U}_{{\sf AP},2}{\bf d}_{{\sf AP}_l,{\sf UE}_j}^{(l-1)}\right)
	\end{aligned}
\end{equation}
where
${\bf Q}_{{\sf AP},1}$,
${\bf U}_{{\sf AP},1}$ and ${\bf U}_{{\sf AP},2}$ are trainable weight matrices.

{\textbf {Updating representations of PS-UE edges:}}
The neighboring edges of a PS-UE edge include both AP-UE and PS-UE edges. Since the random features in $\boldsymbol{\Lambda}$ play a crucial role in distinguishing the actions defined on the  PS-UE edges, the information from all neighboring edges is aggregated.

To meet condition (i), three processing functions are used, parameterized by different weight matrices: the first is for neighboring AP-UE edges, the second is for neighboring PS-UE edges connecting to a PS-vertex, and the third is for neighboring PS-UE edges connecting to a UE-vertex. Since the same weight matrix is used to process the information from different UEs (i.e., the representations of the neighboring PS-UE edges connected to different UEs), such parameter-sharing struggles to capture pilot contamination that varies across UEs. Specifically, if the pilot contamination between two UEs using the same PS is weak, the UEs can share a PS to reduce pilot overhead. Otherwise, they should use different PSs. To improve learning performance, we introduce attention scores that can reflect the potential pilot contamination for processing the information from different UEs.

To meet condition (ii), all PS-UE edges use the same weight matrix for combination.
Then, the hidden representation of edge $({\sf PS}_{g},{\sf UE}_k)$ is updated  in the $l$th layer as
\begin{equation}\label{eq:update-PT-UE-edges}
	\begin{aligned}
		\!\!\!\!\!{\bf d}_{{\sf PS}_g,{\sf UE}_k}^{(l)}\!\!\!&=\phi\Big({\bf Q}_{{\sf PS},1}{\bf d}_{{\sf PS}_g,{\sf UE}_k}^{(l-1)}\!\!+\!\frac{1}{M}\!\!\sum\nolimits_{m=1}^M\!\!{\bf U}_{{\sf PS},1}{\bf d}_{{\sf AP}_m,{\sf UE}_k}^{(l-1)}\\
		&+\!\frac{1}{K}\!\sum\nolimits_{\substack{i=1,i\neq g}}^{K}{\bf U}_{{\sf PS},2}{\bf d}_{{\sf PS}_i,{\sf UE}_k}^{(l-1)}\\
		&+\!\frac{1}{K}\!\sum\nolimits_{\substack{j=1,j\neq k}}^K\left({\bf c}_{jk}^{(l)}\odot({\bf U}_{{\sf PS},3}{\bf d}_{{\sf PS}_g,{\sf UE}_j}^{(l-1)})\right)\Big),
	\end{aligned}
\end{equation}
where ${\bf c}_{jk}^{(l)}$ is an attention score designed as follows
\begin{equation}\label{eq:attention}
	\begin{aligned}
		{\bf c}_{jk}^{(l)} = \phi_{\sf A}\Big(\frac{1}{M}\sum\nolimits_{m=1}^M {({\bf U}_{{\sf PS},4}}{\bf d}_{{\sf AP}_m,{\sf UE}_j}^{(l - 1)} ) \\\odot ({\bf U}_{{\sf PS},5}{\bf d}_{{\sf AP}_m,{\sf UE}_k}^{(l - 1)})\Big),
	\end{aligned}
\end{equation}
${\bf Q}_{{\sf PS},1}$, ${\bf U}_{{\sf PS},1}\sim{\bf U}_{{\sf PS},5}$ are trainable weight matrices, and $\phi_{\sf A}(\cdot)$ is an activation function.

\begin{Remark}
The attention score can be interpreted by considering the hidden representation in the first layer, where ${\bf d}_{{\sf AP}_m,{\sf UE}_j}^{(0)}=[\beta_{mj}, a_{mj}]^{\sf T}$,
and ${\bf U}_{{\sf PS},4}={\bf U}_{{\sf PS},5}=[1,0]$. After omitting $\phi_{\sf A}(\cdot)$, we have ${c}_{j  k}^{(1)}=\frac{1}{M}{\boldsymbol{\beta}}_j^{\sf T}{\boldsymbol{\beta}_k}$, where ${\boldsymbol{\beta}_j}\triangleq[\beta_{1j},\cdots,\beta_{Mj}]^{\sf T}$ is the LSF channel vector of ${\sf UE}_j$. Clearly, ${c}_{j  k}^{(1)}$ reflects the similarity between LSF channel vectors of ${\sf UE}_j$ and ${\sf UE}_k$. A higher similarity indicates more severe pilot contamination if the two UEs use the same PS, as analyzed in the sequel.
Given a pilot assignment matrix ${\bf X}=[{\bf x}_1,\cdots,{\bf x}_K]$, the normalized mean square error of ${\hat{\bf h}}_{mk}$ when ${\bf h}_{mk}\sim{\mathcal CN}({\bf 0},\beta_{mk}{\bf I}_N)$ is \cite{foundation_CF_2021},
\begin{equation} \label{eq:NMSE}
	\!\!{\sf NMSE}_{mk}\!=\!1\!-\!\frac{\beta_{mk}^2}{\sum_{j=1}^K({\bf x}_k^{\sf T}{\bf x}_j){\beta_{mk}}{\beta_{mj}}\!+\!{\sigma_{{\sf AP},m}^2\beta_{mk}}/{p^{\sf ul}}}.
\end{equation}
The term ${\beta_{mk}}{\beta_{mj}}$ reflects the impact of pilot reuse between ${\sf UE}_j$ and ${\sf UE}_k$ on the estimation error. Larger values of ${\beta_{mk}}{\beta_{mj}}$ (and thus higher ${\boldsymbol{\beta}}_j^{\sf T}{\boldsymbol{\beta}_k}$) lead to larger estimation error, indicating more severe pilot contamination.
Thus, the attention score can capture the level of pilot contamination.
\end{Remark}

The pilot assignment matrix yielded by the DTS-Pilot-GNN is the processed representations of PS-UE edges in the $L$th layer
for satisfying the constraints in \eqref{eq:cons-1}. By using softmax function for processing, the action of edge $({\sf PS}_g,{\sf UE}_k)$ is
\begin{equation}\label{eq:pilot-cons}
	{\hat x}_{{\sf DTS},gk}=\exp({d}_{{\sf PS}_g,{\sf UE}_k}^{(L)})/\sum\nolimits_{i=1}^{K}\exp({d}_{{\sf PS}_i,{\sf UE}_k}^{(L)}).
\end{equation}

It can be shown that the policy learned by the DTS-Pilot-GNN has the same permutation property as $f_{\sf DTS}^{\sf pilot}(\cdot)$ after introducing the attention mechanism.

\vspace{-2mm} \subsection{Designing STS-GNN}
This GNN learns the STS policy $[{\bf X}_{\sf STS},{\bf P}_{\sf STS}]=f_{\sf STS}({\boldsymbol\beta},{\bf A})$. It outputs the pilot assignment matrix $\hat{\bf X}_{\sf STS}=[{\hat x}_{{\sf STS},gk}]_{K\times K}$ and the power allocation matrix ${\bf\hat P}_{\sf STS}=[{\hat p}_{{\sf STS},mk}]_{M\times K}$ simultaneously in a frame.
\subsubsection{Graph construction}
Same as the pilot assignment policy, the STS policy is defined on three normal sets: the AP-, UE-, and PS-sets. Hence, the graph constructed for the STS policy has the same topology (including vertices and edges as well as their types) and the same features as the graph for the pilot assignment policy illustrated in Fig. \ref{fig:graph_LSF}. The actions for pilot assignment are defined on PS-UE edges, where the action on edge $({\sf PS}_g,{\sf UE}_k)$ is ${\hat x}_{{\sf STS},gk}$. The actions for power allocation are defined on AP-UE edges, where the action on edge $({\sf AP}_m,{\sf UE}_k)$ is ${\hat p}_{{\sf STS},mk}$.

\subsubsection{Feature enhancement}\label{sec:FE-1}
With the same analysis as in Sec. \ref{sec:FE}, artificial features need to be introduced into PS-UE edges for distinguishing their representations. Hence, the input of the designed GNN is ${\boldsymbol\beta}$, ${\bf A}$, and $\boldsymbol{\Lambda}$.

\subsubsection{Updating equations design}
Since the parameter-sharing in a GNN depends on the graph topology, the hidden representations of the AP-UE edges and PS-UE edges are also updated with  \eqref{eq:update-AP-UE-edges} and \eqref{eq:update-PT-UE-edges}, respectively. The only difference from the DTS-Pilot-GNN is: the STS-GNN has two kinds of actions, where the pilot assignment matrix is obtained with \eqref{eq:pilot-cons}, and the power allocation matrix is obtained by applying the output-layer activation function from \cite{DataRateModel_GJ2023}.


\subsection{Training and Test Phases}

\subsubsection{Training phase}
Considering that the LSF and SSF channels change in two time scales, each sample contains $N_{\sf T}$ subsamples that share the same LSF channel but have different SSF channels. We train the GNNs via unsupervised learning.

To maximize ${\mathbb E}_{{\bf H}^t}\{\eta^t\}$, $\sum\nolimits_{i=1}^{N_{\sf T}}\eta_{n,i}/{N_{\sf T}}$ is used as part of the loss function for training, where $\eta_{n,i}$ is the net-SE corresponding to the $ i$th subsample of the $n$th sample.
To facilitate backpropagation during training, the GNN for pilot assignment outputs the probabilities of assigning PSs to UEs rather than binary values. To make the probabilities close to 0 or 1, the loss function includes a penalty term. Then, the loss function is
\begin{equation}\label{eq:loss_function}
	\begin{aligned}
	{\mathcal L}&=-\frac{1}{N_s}\sum\nolimits_{n=1}^{N_{\sf s}}\Big(\frac{1}{N_{\sf T}}\sum\nolimits_{i=1}^{N_{\sf T}}\eta_{n,i}\\
	&-w\sum\nolimits_{g = 1}^{K} \sum\nolimits_{k = 1}^K {\log_e {\hat x}_{gk,n}}{\log_e (1 -{\hat x}_{gk,n})}\Big),
	\end{aligned}
\end{equation}
where $N_{\sf s}$ is the number of samples in a batch,
${\hat x}_{gk,n}$ is the probability of assigning ${\sf PS}_g$ to ${\sf UE}_k$ for the $n$th sample, and $w$ is a hyper-parameter.

To jointly train the DTS-Pilot-GNN and DTS-Power-GNN, we propose a DTS framework. As shown in Fig. \ref{fig:framework-problem}(a), the channel estimation and beamforming modules are used to compute the equivalent channel gains in \eqref{eq:estimated-equivalent-channel}, which are the input of the DTS-Power-GNN. Multiple DTS-Power-GNNs with the same weight matrices are used to output the allocated powers for the $N_{\sf T}$ subsamples. The trainable weight matrices in the DTS-Pilot-GNN and DTS-Power-GNN are updated synchronously through backpropagation.

To train the STS-GNN in Fig. \ref{fig:framework-problem}(b), the beamforming
vectors are required to compute the loss function in \eqref{eq:loss_function} for backpropagation. This, in turn, necessitates the use of channel estimation and beamforming modules to compute these
vectors based on one of the outputs of the
STS-GNN (i.e., $\hat{\bf X}_{\sf STS}$).

\begin{figure}[!thp]
	\centering
	\includegraphics[scale=0.53]{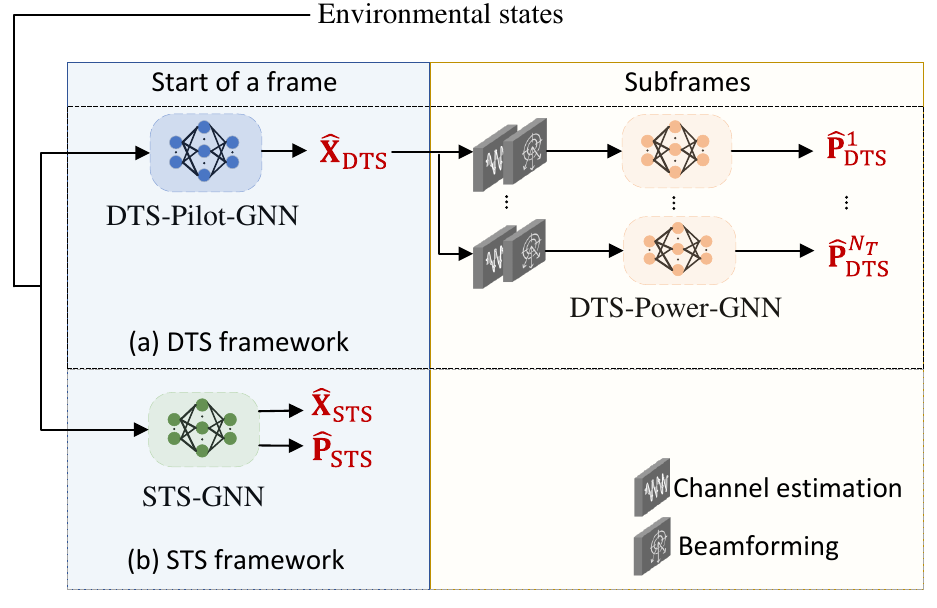}
	\vspace{-2mm}
	\caption{The DTS and STS frameworks.} \label{fig:framework-problem}
\vspace{-6mm}
\end{figure}

\subsubsection{Test phase}
The pilot assignment probability matrices output by the DTS-Pilot-GNN and STS-GNN are discretized by setting the maximal value in each column as 1 and the remaining values as 0. From these discretized matrices, we obtain $\mathbf{X}_{\sf o}\in\{0,1\}^{\tau_{\sf p}\times K}$ by removing all-zero rows, which simultaneously determines the pilot length $\tau_{\sf p}$.

As shown in Fig. \ref{fig:framework-problem}, the DTS-Pilot-GNN is executed at the start of a frame, and the DTS-Power-GNN is executed in each subframe. The STS-GNN is executed at the start of a frame.

The proposed DTS and STS frameworks include attention-based GNNs, which are called DTS-AGNN and STS-AGNN for short, respectively.

\section{Simulation Results}\label{sec:Result}
In this section, we evaluate the learning performance, generalizability, training complexity, and inference time of the proposed frameworks.

\begin{table}[htbp]
	\centering
	\caption{Simulation Setup \cite{2024_3GPP_901}}\label{table:setup}
	\begin{tabular}{c|c|c}
		\hline\hline
		Description&Notation& Value\\
		\hline
		Carrier frequency&$f_c$&6 $\sf GHz$\\
		\hline
		Bandwidth&$B$&20 $\sf MHz$\\
		\hline
		\makecell[c]{Default number of\\time slots in a subframe}&$\tau_{\sf c}$&200\\
		\hline
		\makecell[c]{Number of subframes}&$N_{\sf T}$&10\\
		\hline
		Uplink transmit power&$p^{\sf ul}$&23 $\sf dBm$\\
		\hline
		\makecell[c]{Default maximal\\ downlink transmit power} &$P_{\sf max}$&44 $\sf dBm$\\
		\hline
		Noise spectral density&$N_0$&-174 $\sf dBm/Hz$\\
		\hline
		Noise figure&$N_{\sf F}$&9 $\sf dB$\\
		\hline
		Height of AP antennas&-&10 $\sf m$\\
		\hline
		Height of UE antennas&-&1.5 $\sf m$\\
		\hline
		Default ISD &-&200 $\sf m$\\
		\hline
		\makecell[c]{Minimum distance\\ between AP and UE}&-&10 $\sf m$\\
		\hline\hline
	\end{tabular}
\end{table}

Consider a CF-MIMO system where the intersite distance (ISD) between APs is fixed and UEs are randomly placed. We consider an urban micro-cell (UMi) scenario. The channel model consists of LSF and Rayleigh fading. The LSF model is ${\beta}_{mk} ({\sf dB})=-32.4-20\log_{10}(f_c({\sf GHz}))-31.9\log_{10}(s_{mk}({\sf m}))+\chi$, where $s_{mk}$ is the distance between ${\sf AP}_m$ and ${\sf UE}_k$, and $\chi$ is the shadowing with a standard deviation of $8.2$ $\sf dB$ \cite{2024_3GPP_901}. The noise power is $\sigma_{{\sf AP},l}^2({\sf dBm})\!=\!\sigma_{{\sf UE},k}^2({\sf dBm})\!=\!N_0+10\log_{10}(B)+N_{\sf F}$. The simulation setup is provided in Table \ref{table:setup}.

The association matrix $\bf A$ is determined as follows. Each UE (e.g., ${\sf UE}_k$) first associates to the AP with the highest LSF channel gain, and then associates to ${\sf AP}_m$ if ${\beta}_{mk}\geq \rho$, where $\rho$ is a threshold that is tuned to maximize the net-SE. In the UMi scenario, $\rho$ is set as $\rho ({\sf dB})=-32.4-20\log_{10}(6)-31.9\log_{10}(200)$. We adopt the distributed regularized zero-forcing beamforming in \cite{foundation_CF_2021}, i.e., ${{\bf{w}}_{{m_k}}^{t}} = {({\sum _{i \in {{\cal A}_m}}}\!\!{{\hat{\bf{h}}}_{mi}^{t}}({\hat{\bf{h}}_{mi}^{t}})^{\rm{H}} \!\!+\!\! {\sum _{j \notin {{\cal A}_m}}}{\beta _{mj}}{{\bf{I}}_N}\!\!+\!\! \frac{\sigma_{{\sf UE},k}^2}{p^{\sf ul}}{{\bf{I}}_N})^{ - 1}}{\hat{\bf{h}}_{mk}^{t}}$.

5,000 samples are generated for training, and 100 samples are used for testing.

These settings are considered unless otherwise specified.

The optimizer is Adam, and batch normalization is applied. After fine-tuning, the elements in $\boldsymbol{\Lambda}$ are drawn from a uniform distribution over $[0, 1)$, and the activation functions of the combination functions in the hidden layers and the attention mechanism are ${\sf relu}(\cdot)$ and ${\sf tanh}(\cdot)$, respectively. Other fine-tuned hyper-parameters are as follows. The initial learning rate is 0.01. The sizes of hidden representation vectors are $[8,8,8,8,8,8,1]$ for the seven layers in DTS-Pilot-GNN, $[16,16,16,2]$ for the four layers in DTS-Power-GNN, and $[8,8,8,8,8,8,2]$ for the seven layers in STS-GNN. The parameters in the loss function are $N_{\sf s}=50$ and $w=0.2$.

All simulations are conducted on a computer equipped with one Intel i9-10940X CPU and one Nvidia RTX 3080Ti GPU.

\vspace{-2mm}
\subsection{Learning Performance} \label{sec:learning-performance}
The learning performance is measured by the average net-SE in \eqref{eq:Avg-net-SE} (hereafter referred to as net-SE for short).

We compare the net-SE of the proposed frameworks with those of the baselines, which include a numerical algorithm and three learning-based methods.

Since no existing algorithms can solve the joint optimization problem (which optimizes the pilot length) that we aim to address, we develop a numerical baseline by combining existing algorithms.
Specifically, the Dsatur algorithm in \cite{2016Dsatur} is used to generate the initial pilot assignment, and then the Tabu-based algorithm in \cite{2020Tabu} is used for reassignment with a tabu length of $K$ and a maximum of $10$ iterations. This combination harnesses the strengths of both algorithms: the Dsatur algorithm can adjust pilot length, while the Tabu-based algorithm can maximize the SE with a given pilot length.
Once the pilot assignment is determined, the weighted MMSE (WMMSE) algorithm in \cite{Imperfect_CSI_2019} is adopted for power allocation. This numerical baseline is named ``Dsatur-tabu+WMMSE".

Three learning-based baselines are listed below.

\begin{itemize}
    \item STS-AGNN-Fix ($\tau_{\sf p}=z$): The same as existing learning-based methods in \cite{2024JPCPA-CF-Khan,JPAPC2025}, this baseline solve the STS problem given predetermined pilot length $z$. For a fair comparison, the same STS framework is used as the proposed STS-AGNN, which however is trained on a graph containing $z$ PS-vertices.
         This baseline is used to show the performance gain from optimizing pilot length.
    \item STS-AGNN w/o FE: The STS framework that includes the GNN without feature-enhancement (FE), i.e., the inputs of the STS-AGNN do not include $\boldsymbol{\Lambda}$.
        This baseline is used to show the necessity of FE.
    \item DTS/STS-GNN: The DTS and STS frameworks that include GNNs without attention.
        This baseline is used to evaluate the contribution of the attention mechanism.
\end{itemize}

\subsubsection{Ablation experiments}
In Fig. \ref{fig:perf_tau}, we show the impact of the pilot-length optimization, feature enhancement, and attention mechanism on the net-SE by comparing the proposed DTS/STS-AGNN with the learning-based baselines. The observations from the results are listed as follows.

\textbf{Pilot-length optimization:}
 The net-SE of ``STS-AGNN-Fix ($\tau_{\sf p}=32$)" is lower than ``STS-AGNN-Fix ($\tau_{\sf p}=z$)" with $z=\{10,18\}$. ``STS-AGNN-Fix ($\tau_{\sf p}=10$)" performs better when $\tau_{\sf c}$ is small, while ``STS-AGNN-Fix ($\tau_{\sf p}=18$)" performs better when $\tau_{\sf c}$ is large. These results demonstrate that the pilot length should be dynamically adjusted according to the channel coherence time. ``STS-AGNN" consistently achieves higher net-SE than ``STS-AGNN-Fix ($\tau_{\sf p}=z$)" with $z=\{10,18,32\}$. These performance gains underscore the advantage of optimizing pilot length.

\textbf{Feature enhancement:}
``STS-AGNN" achieves significantly higher net-SE than ``STS-AGNN w/o FE", owing to the ability of distinguishing the PS-UE edges.

\textbf{Attention mechanism:}
By comparing ``DTS-AGNN" with ``DTS-GNN"  and ``STS-AGNN" with ``STS-GNN", we can observe the net-SE gain from the attention mechanism.

\vspace{-3mm}
\begin{figure}[!htp]
	\centering
	\includegraphics[width=0.38\textwidth]{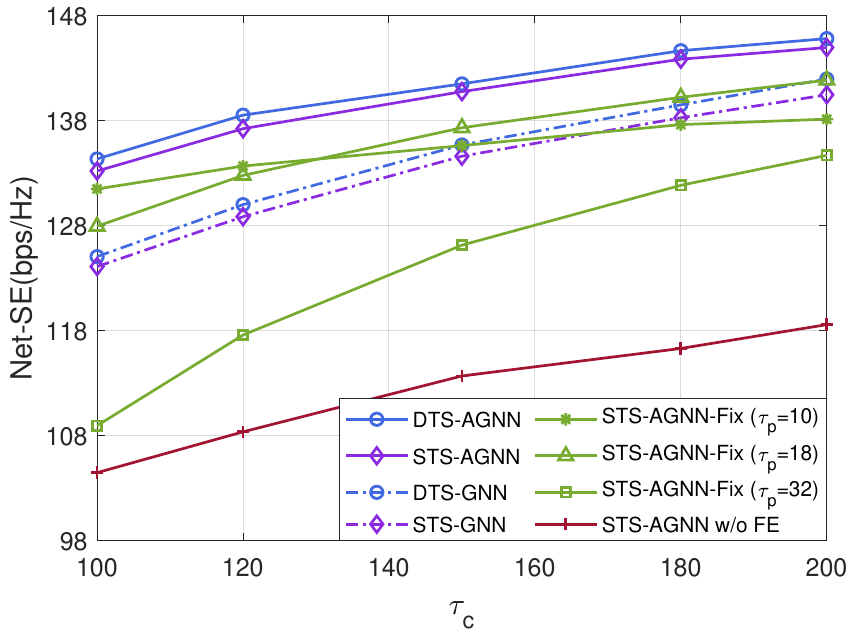}
	\vspace{-2mm}
	\caption{Learning performance, $M=7$, $N=8$, $K=35$.} \label{fig:perf_tau}
\end{figure}
\vspace{-6mm}
\begin{figure}[!htp]
	\centering
	\includegraphics[width=0.38\textwidth]{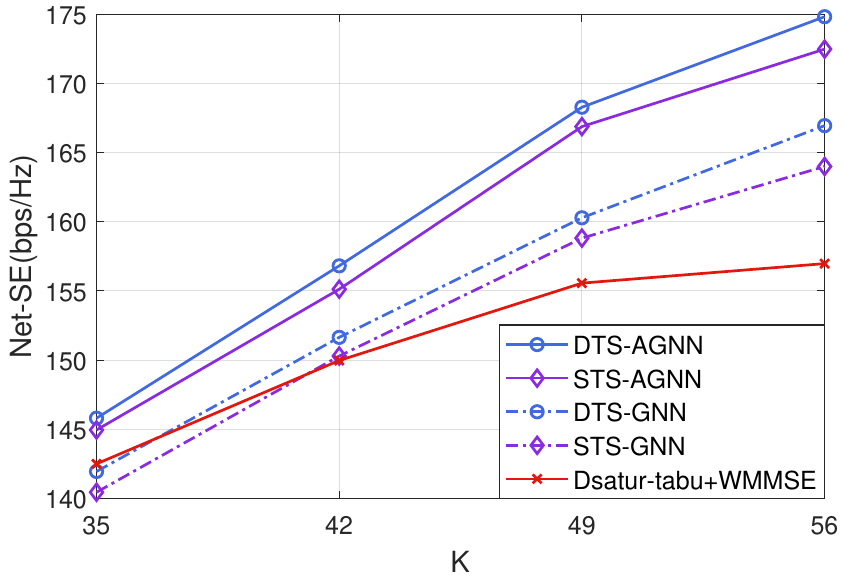}
\vspace{-2mm}
	\caption{Impact of the number of UEs, $M=7$, $N=8$.} \label{fig:perf_UE}
\vspace{-3mm}
\end{figure}

\subsubsection{Impact of the number of UEs}\label{sec:Varying_UEs}
In Fig. \ref{fig:perf_UE}, we provide the net-SE achieved by the GNNs and the numerical algorithm.
We can see that ``DTS/STS-AGNN" achieves higher net-SE than ``Dsatur-tabu+WMMSE", and the performance gap grows with $K$. The main reason is that ``Dsatur-tabu+WMMSE'' separately optimizes pilot length, pilot assignment, and power allocation.

The performance gap between ``DTS/STS-AGNN" and ``DTS/STS-GNN" increases with $K$. This is because the attention mechanism can reflect the pilot contamination, which becomes more severe for more UEs with given $M$ and $N$.

\subsubsection{Impact of the number of APs}
In Fig. \ref{fig:perf_AP}, we compare the GNNs and the numerical algorithm with different numbers of APs given a constant $K/M$. We can see that the performance gap between ``DTS/STS-AGNN" and ``DTS/STS-GNN" grows with $M$. It is because the growth of $K$ with $M$ requires intensified pilot reuse among UEs to reduce the pilot overhead. ``DTS/STS-AGNN" captures pilot contamination through the attention mechanism, enabling appropriate pilot reuse among UEs. In contrast, ``DTS/STS-GNN" lacks this mechanism, resulting in a larger number of assigned PSs and hence higher pilot overhead when $M$ is larger.
\vspace{-4mm}
\begin{figure}[!htp]
	\centering
	\includegraphics[width=0.38\textwidth]{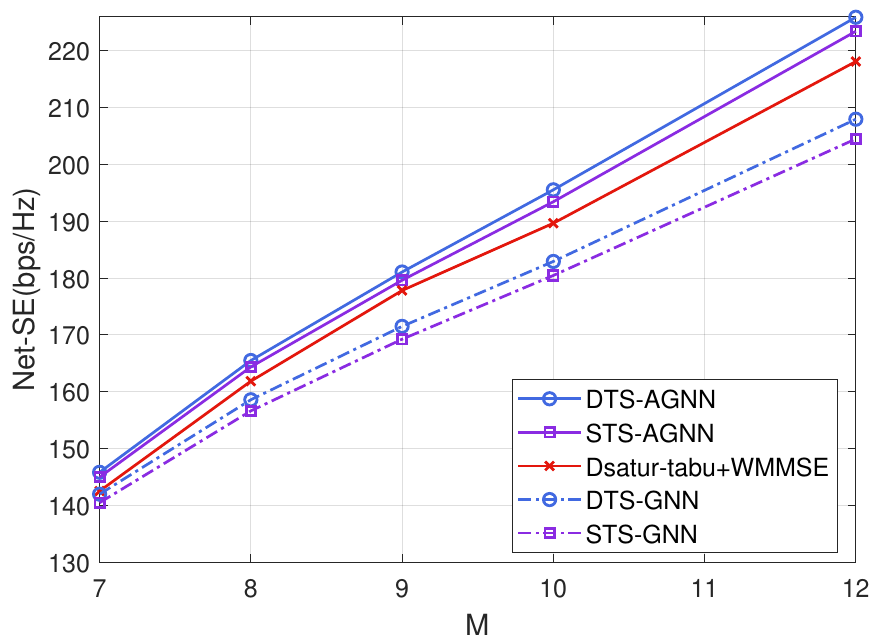}
\vspace{-2mm}
	\caption{Impact of the number of APs, $N=8$, $K/M=5$.} \label{fig:perf_AP}
    \vspace{-3mm}
\end{figure}

\subsubsection{Impact of the number of antennas at each AP}
In Table \ref{table:perf_N}, we show the impact of $N$.
We can see that ``DTS-AGNN" consistently outperforms ``STS-AGNN", but the gain is evident only when $N$ is small. When $N=2$, the net-SE achieved by ``STS-AGNN" falls below that of the numerical baseline and is only 92.35\% of the net-SE achieved by ``DTS-AGNN". This is because a large antenna array causes channel hardening, making power allocation primarily dependent on LSF channels.

\begin{table}[htbp]
	\centering
    \small
	\vspace{-3mm}
	\caption{Impact of $N$, $M=7$, $K=35$.}\label{table:perf_N}
	\begin{tabular}{c|c|c|c|c}
		\hline\hline
		{\diagbox[width=11em]{Net-SE(bps/Hz)}{$N$}}& 2&4 &8&16\\
		\hline
		DTS-AGNN&64.76&97.59&145.80&205.11\\
		\hline
		STS-AGNN&59.65&94.21&144.95&204.29\\
        \hline
        Dsater-tabu+WMMSE&62.87&93.46&142.50&203.95\\
		\hline\hline
	\end{tabular}
	\vspace{-4mm}
\end{table}

\vspace{-3mm}
\subsection{Generalizability to Unseen Scales and Channel Model}\label{sec:generalize}
Given that the learning-based baselines cannot perform well, we only assess the generalizability of the proposed DTS-AGNN and STS-AGNN. Since ``DTS-AGNN" achieves the highest net-SE, the generalization performance of a GNN is measured by a $\textbf {net-SE ratio}$, defined as $\frac{\text{net-SE of a GNN trained on a scenario}}{\text{net-SE of DTS-AGNN trained on the test scenario}}$.

In Table \ref{table:generalize_UE}, we provide the net-SE ratios of the GNNs trained on the scenario where $K=42$ and tested on the scenarios where $K$ ranges from 35 to 56. The results show that ``DTS/STS-AGNN" is well generalized across $K$.

\begin{table}[htbp]
	\centering
    \small
	\vspace{-2mm}
	\caption{Generalizability to $K$, $M=7$, $N=8$}\label{table:generalize_UE}
	\begin{tabular}{c|c|c|c}
		\hline\hline
		{$K$}& 35 &49&56\\
		\hline
		DTS-AGNN&98.93\%&98.87\%&97.22\%\\
		\hline
		STS-AGNN&97.49\%&98.26\%&96.57\%\\
		\hline\hline
	\end{tabular}
\end{table}

In Table \ref{table:generalize_N}, we present the net-SE ratios of the GNNs trained on the scenario where $N=16$ (with legend ``DTS/STS-AGNN ($N=16$)") and tested on the scenarios where $N$ ranges from 4 to 64. We can see that these GNNs can be well generalized to $N$ when $N \geq 8$. This is because the pilot assignment policies differ for $N=16$ and $N=4$.
When $N=16$, each AP is equipped with more antennas than its associated UEs. In this case, since an AP can serve all associated UEs, different PSs can be assigned to these UEs to avoid mutual pilot contamination, thereby mitigating intra-AP interference during data transmission.
In contrast, when $N=4$, each AP has fewer antennas than its associated UEs and thus serves only a subset of them through power allocation. Since the UEs associated with but not served by an AP (possibly served by other APs) do not cause intra-AP interference during data transmission, pilot reuse can be introduced among these UEs to reduce the number of assigned PSs and thereby reduce the pilot overhead.

\begin{table}[htbp]
	\centering
    \small
	\caption{Generalizability to $N$, $M=7$, $K=35$}\label{table:generalize_N}
	\resizebox{0.48\textwidth}{!}{
	\begin{tabular}{c|c|c|c|c|c}
		\hline\hline
		{\diagbox[width=7em]{\makecell[l]{Net-SE ratio}}{$N$}}& 4&8&16&32&64\\
		\hline
		\makecell[c]{DTS-AGNN\\($N=16$)} &87.54\%&96.71\%&100\%&99.90\%&98.77\%\\
		\hline
		\makecell[c]{STS-AGNN\\($N=16$)}&76.78\%&94.42\%&99.37\%&99.47\%&98.73\%\\
		\hline
		\makecell[c]{DTS-AGNN\\($N=8,16$)} &93.74\%&99.07\%&99.31\%&97.47\%&95.78\%\\
		\hline
		\makecell[c]{STS-AGNN\\($N=8,16$)}&82.02\%&97.57\%&98.50\%&96.00\%&94.00\%\\	
        \hline\hline
	\end{tabular}}
\end{table}

To improve the generalization performance for $N=4$, we train the GNNs on the scenarios where $N=8,~16$ (with legend ``DTS/STS-AGNN ($N=8,16$)"). We can see that
``DTS-AGNN" can be well generalized to $N=4$, but ``STS-AGNN" cannot. It is because the power allocation policy is highly sensitive to SSF channels when $N$ is small, whereas ``STS-AGNN" only has LSF channel gains as input.

In Table \ref{table:generalize_AP}, we provide the net-SE ratios of the GNNs trained on the UMi scenario where $M=7$ and tested on the UMi scenarios where $M$ varies from 8 to 12  (with legend ``DTS/STS-AGNN (UMi)").
To evaluate the generalization performance of GNNs across channel models, we also provide the net-SE ratios achieved by the GNNs trained in the UMi scenario with $M=7$ but tested in the urban macro-cell (UMa) scenario (with legend ``DTS/STS-AGNN (UMa)"). In the UMa scenario, the LSF channel model is ${\beta}_{mk} ({\sf dB})=-32.4-20\log_{10}(f_c({\sf GHz}))-30\log_{10}(s_{mk}({\sf m}))+\chi$ where the shadowing has a standard deviation of $7.8$ dB \cite{2024_3GPP_901}, the maximal downlink transmit power is $49~{\sf dBm}$, the AP antenna height is $25~{\sf m}$, the ISD is $500~{\sf m}$, the minimal AP-UE distance is $35~{\sf m}$, and the AP-UE association threshold is $\rho ({\sf dB})=-32.4-20\log_{10}(6)-30\log_{10}(450)$.
The results show that the GNNs can be well generalized to unseen numbers of APs and to the UMa scenario.

\begin{table}[htbp]
	\centering
    \small
	\caption{Generalizability to $M$ and channel model, $N=8$, $\frac{K}{M}=5$}\label{table:generalize_AP}
	\begin{tabular}{c|c|c|c|c}
		\hline\hline
		{\diagbox[width=10em]{Net-SE ratio}{$M$}}& 8&9&10&12\\
		\hline
		DTS-AGNN (UMi)&99.60\%&99.04\%&97.28\%&95.39\%\\
		\hline
		STS-AGNN (UMi)&99.24\%&98.82\%&96.98\%&94.71\%\\
        \hline
		DTS-AGNN (UMa)&98.84\%&98.12\%&97.12\%&95.34\%\\
		\hline
		STS-AGNN (UMa)&98.14\%&97.32\%&96.44\%&94.25\%\\
		\hline\hline
	\end{tabular}
	\vspace{-5mm}
\end{table}

\subsection{Training Complexity and Inference Time}
In Fig. \ref{fig:sample}, we provide the net-SE achieved by the GNNs trained with different numbers of samples. We can see that both ``DTS-AGNN" and ``STS-AGNN" need much fewer samples to achieve the performance of the numerical algorithm ``Dsatur-Tabu+WMMSE" (referred to as the baseline performance) than ``DTS-GNN". This indicates that the attention mechanism can reduce the training samples to achieve the baseline performance.
\begin{figure}[!htp]
	\centering
	\includegraphics[width=0.48\textwidth]{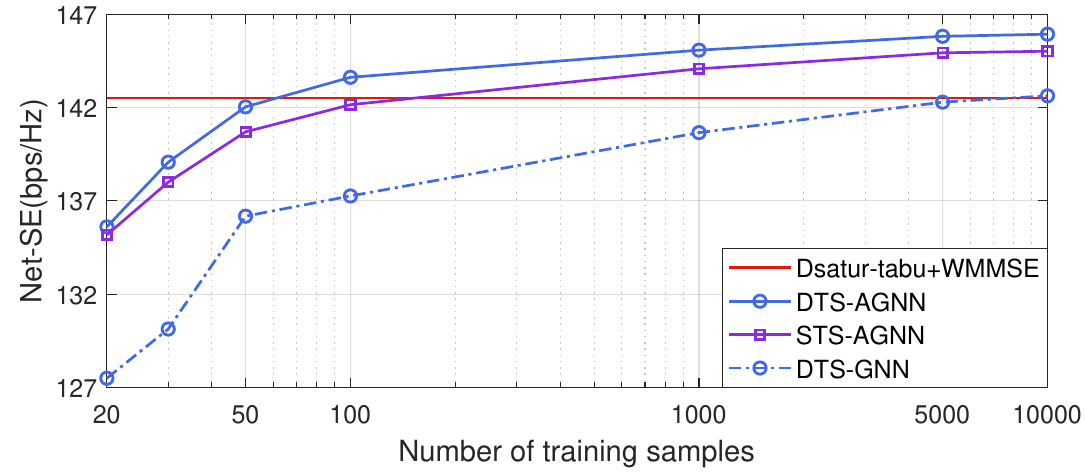}
	\caption{The net-SE achieved by GNNs and the numerical algorithm, $M=7$, $N=8$, $K=35$.} \label{fig:sample}
\end{figure}

In Table \ref{table: train_complexity}, we provide the training complexity of the GNNs in terms of sample, space, and time complexity. The sample complexity is the minimum number of samples required to achieve the baseline performance. The space complexity is the number of trainable parameters of the fine-tuned GNNs. The time complexity is the training time on GPU required to achieve the baseline performance with the minimal number of samples. We can see that ``DTS-AGNN" requires the fewest training samples and shortest training time. However, its space complexity is higher than ``STS-AGNN" because it consists of two GNNs, whereas ``STS-AGNN" contains only a single GNN. Without the attention mechanism, ``DTS-GNN" needs much more training samples and trainable parameters, as well as much longer training time than the other two GNNs.

\begin{table}[htbp]
	\centering
	\small
	\vspace{-1mm}
	\caption{Training complexity ($M=7$, $N=8$, $K=35$)}\label{table: train_complexity}
	\begin{tabular}{c|c|c|c}
		\hline\hline
		GNNs& Sample&Space &\makecell[c]{Training time}\\
		\hline
		DTS-AGNN&90&3012&41.89 s\\
		\hline
		STS-AGNN&170&2652&49.37 s\\
		\hline
		DTS-GNN&8000&17180&29.94 min\\
		\hline\hline
	\end{tabular}
\end{table}

In Table \ref{table: inference_time}, we provide the total inference time of the GNNs and the numerical algorithm in a frame. Since the numerical algorithm is usually executed on a CPU, we also evaluate the inference time of the GNNs on CPU for a fair comparison. It is shown that ``Dsatur-Tabu+WMMSE" requires dramatically longer running time than the GNNs.
``STS-AGNN" has the shortest inference time. It is because the STS-AGNN is executed once, whereas the DTS-Power-GNN in the DTS framework is executed $N_{\sf T}$ times, as shown in Fig. \ref{fig:framework-problem}. Since the inference of GNNs can be executed in parallel \cite{CPY2025}, the inference time of the GNNs is also evaluated on the GPU, which is much shorter than that on the CPU.

\begin{table}[htbp]
	\centering
	\vspace{-1mm}
    \small
	\caption{Inference time (millisecond) ($M=7$, $N=8$)}\label{table: inference_time}
	\begin{tabular}{c|c|c|c|c}
		\hline\hline
		\multicolumn{2}{c|}{\diagbox[width=15em]{Methods}{$K$}}&{35}&
		{42}&{49}\\
		\hline
		\multirow{2}{*}{DTS-AGNN}&CPU&48.7&56.5&116.9\\
		\cline{2-5}
		&GPU&28.9&29.5&30.2\\
		\hline
		\multirow{2}{*}{STS-AGNN}&CPU&20.5&25.2&33.5\\
		\cline{2-5}
		&GPU&9.8&10.3&11.1\\
		\hline
		\multirow{2}{*}{DTS-GNN}&CPU&122.5&155.8&189.2\\
		\cline{2-5}
		&GPU&34.2&34.9&36.4\\
		\hline
		Dsatur-Tabu+WMMSE&CPU&246381&315613&376011\\
		\hline\hline
	\end{tabular}
\end{table}

\section{Conclusion}\label{sec:Conclusion}
In this paper, we developed GNN-based frameworks for the joint optimization of pilot length, pilot assignment, and power allocation in both DTS and STS scenarios. The proposed GNNs satisfy the permutation properties of optimal policies, which not only need low training complexity but also support size generalizability. Notably, owing to their generalizability to different numbers of UEs, the GNNs can output pilot assignment matrices of variable size.
For the GNNs dedicated to pilot assignment, both the integrated feature enhancement technique and the contamination-aware attention mechanism boost the learning performance while reducing the training complexity. Simulation results demonstrated that the proposed DTS-AGNN and STS-AGNN consistently outperform baseline methods in terms of net-SE, training complexity, and inference latency. Meanwhile, the DTS-AGNN and STS-AGNN can be well generalized across different problem scales and channel models. Besides, the STS-AGNN performs closely to the DTS-AGNN as the number of antennas increases, with lower space complexity and inference latency.


\bibliography{ref}

@ARTICLE{Zhou2023,
  author={Zhou, Xingguang and Xia, Wenchao and Zhang, Jun and Wen, Wanli and Zhu, Hongbo},
  journal={IEEE Trans. Wireless Commun.}, 
  title={Joint Optimization of Frame Structure and Power Allocation for {URLLC} in Short Blocklength Regime}, 
  year={2023},
  volume={71},
  number={12},
  pages={7333-7346},
  month={Dec.},}

@ARTICLE{Cheng2017,
  author={Cheng, Hei Victor and Björnson, Emil and Larsson, Erik G.},
  journal={IEEE Trans. Signal Process.}, 
  title={Optimal Pilot and Payload Power Control in Single-Cell Massive {MIMO} Systems}, 
  year={2017},
  volume={65},
  number={9},
  pages={2363-2378},
  month={May},}

@ARTICLE{Peng2023,
  author={Peng, Qihao and Ren, Hong and Dong, Mianxiong and Elkashlan, Maged and Wong, Kai-Kit and Hanzo, Lajos},
  journal={IEEE J. Sel. Areas Commun.}, 
  title={Resource Allocation for Cell-Free Massive {MIMO}-Aided {URLLC} Systems Relying on Pilot Sharing}, 
  year={2023},
  volume={41},
  number={7},
  pages={2193-2207},
  month={July},}

@ARTICLE{Wu2023,
  author={Wu, Jilong and Sheng, Zheng and Zhu, Pengcheng and Jiang, Yanxiang and Wang, Yan},
  journal={IEEE ICCT}, 
  title={Joint Optimization of Pilot Length and Pilot Allocation for {URLLC} in Cell-free Massive {MIMO} Systems}, 
  year={2023},}

@ARTICLE{LSF_CSI_Power2025,
  author={Mobini, Zahra and Ngo, Hien Quoc},
  journal={IEEE Signal Process. Mag.}, 
  title={Massive Multiple-Input, Multiple-Output: Instantaneous versus statistical channel state information-based power allocation [Lecture Notes]}, 
  year={2025},
  volume={42},
  number={2},
  pages={27-36},
  month={March},}

@ARTICLE{CSIPower2024,
  author={Raghunath, Ramprasad and Peng, Bile and Jorswieck, Eduard A.},
  journal={IEEE ICMLCN}, 
  title={Energy-Efficient Power Allocation in Cell-Free Massive {MIMO} via Graph Neural Networks}, 
  year={2024},}

@ARTICLE{CPY2025,
	author={Cong, Pengyu and Yang, Chenyang and Han, Shengqian and Han, Shuangfeng and Wang, Xiaoyun},
	journal={IEEE Open J. Veh. Technol.}, 
	title={Time Complexity of Training DNNs With Parallel Computing for Wireless Communications}, 
	year={2025},
	volume={6},
	number={},
	month={Jan.},
	pages={359-384},}

@INPROCEEDINGS{JPAPC2025,
  author={Cheng, Mengqian and Yang, Yifei and Chen, Changju and Zhu, Pengcheng},
  booktitle={IEEE WCNC}, 
  title={{AI}-Enabled Joint Pilot Assignment and Power Control for Short Packet Transmission in Cell-Free {mMIMO} System}, 
  year={2025},}

@ARTICLE{ClusteringPA2025,
	author={Zhao, Yu and Zhang, Fengming and Gao, Yangjun and Hu, Gang},
	journal={IEEE Commun. Lett.}, 
	title={A Fast Pilot Assignment for Cell-Free Massive {MIMO}: Using Anchor-Based Clustering Scheme}, 
	year={2025},
	volume={14},
	number={4},
	month={April},
	pages={1109-1113},}

@ARTICLE{2021KCut,
  author={Zeng, Wenbo and He, Yigang and Li, Bing and Wang, Shudong},
  journal={IEEE Trans. Veh. Technol.}, 
  title={Pilot Assignment for Cell Free Massive {MIMO} Systems Using a Weighted Graphic Framework}, 
  year={2021},
  volume={70},
  number={6},
  month={June},
  pages={6190-6194},}

@ARTICLE{2020GraphColoring,
  author={Liu, Heng and Zhang, Jiayi and Jin, Shi and Ai, Bo},
  journal={IEEE Trans. Veh. Technol.}, 
  title={Graph Coloring Based Pilot Assignment for Cell-Free Massive {MIMO} Systems}, 
  year={2020},
  volume={69},
  number={8},
  month={Aug.},
  pages={9180-9184},}

@BOOK{foundation_CF_2021,
  author={Demir, {\"O}zlemTugfe and Björnson, Emil and Sanguinetti, Luca},
  title={Foundations of user-centric cell-Free massive {MIMO}},
  publisher={Now Foundations and Trends},
  year={2021},
  volume={},
  number={},
  pages={},
  doi={}}

@ARTICLE{2020B5G,
	author={Zhang, Jiayi and Björnson, Emil and Matthaiou, Michail and Ng, DerrickWingKwan and Yang, Hong and Love, DavidJ.},
	journal={IEEE J. Sel. Areas Commun.}, 
	title={Prospective Multiple Antenna Technologies for Beyond {5G}}, 
	year={2020},
	volume={38},
	number={8},
    month={Aug.},
	pages={1637-1660},}

@ARTICLE{cellfreeVersus2017,
  author={Ngo, HienQuoc and Ashikhmin, Alexei and Yang, Hong and Larsson, ErikG. and Marzetta, ThomasL.},
  journal={IEEE Trans. Wireless Commun.}, 
  title={Cell-Free Massive {MIMO} Versus Small Cells}, 
  year={2017},
  volume={16},
  number={3},
  month={March},
  pages={1834-1850},}

@ARTICLE{Imperfect_CSI_2019,
  author={Pan, Cunhua and Ren, Hong and Elkashlan, Maged and Nallanathan, Arumugam and Hanzo, Lajos},
  journal={IEEE Trans. Wireless Commun.}, 
  title={Weighted Sum-Rate Maximization for the Ultra-Dense User-Centric {TDD} {C-RAN} Downlink Relying on Imperfect {CSI}}, 
  year={2019},
  volume={18},
  number={2},
  month={Feb.},
  pages={1182-1198},}

@ARTICLE{LSJ2023MGNN,
	author={Liu, Shengjie and Guo, Jia and Yang, Chenyang},
	journal={IEEE Trans. Wireless Commun.},
	title={Multidimensional Graph Neural Networks for Wireless Communications},
	year={2024},
	month={April},
	volume={23},
	number={4},
	pages={3057-3073},
}

@ARTICLE{2024EdgeGNN-Peng,
  author={Peng, Yao and Guo, Jia and Yang, Chenyang},
  journal={IEEE Trans. Mach. Learn. Commun. Netw.}, 
  title={Learning Resource Allocation Policy: Vertex-{GNN} or {E}dge-{GNN}?}, 
  year={2024},
  month={Jan.},
  volume={2},
  number={},
  pages={190-209},}

@ARTICLE{2022UCSurvey-Ammar,
	author={Ammar, Hussein and Adve, Raviraj and Shahbazpanahi, Shahram and Boudreau, Gary and Srinivas, Kothapalli},
	journal={IEEE Commun. Surveys Tuts.}, 
	title={User-Centric Cell-Free Massive {MIMO} Networks: A Survey of Opportunities, Challenges and Solutions}, 
	year={2022},
	month={Dec.},
	volume={24},
	number={1},
	pages={611-652},}

@ARTICLE{2024_3GPP_901,
	author={3GPP},
	journal={TR 38.901, Version 18.0.0},
	title={Study on channel model for frequencies from 0.5 to 100 {GHz}},
	year={2024}
}

@ARTICLE{2020Tabu,
	author={Liu, Heng and Zhang, Jiayi and Zhang, Xiaodan and Kurniawan, Adit and Juhana, Tutun and Ai, Bo},
	journal={IEEE Trans. Veh. Technol.}, 
	title={Tabu-Search-Based Pilot Assignment for Cell-Free Massive {MIMO} Systems}, 
	year={2020},
	month={Feb.},
	volume={69},
	number={2},
	pages={2286-2290},}

@ARTICLE{2016Dsatur,
	author={Chen, Zhilin and Hou, Xueying and Yang, Chenyang},
	journal={IEEE Trans. Veh. Technol.}, 
	title={Training Resource Allocation for User-Centric Base Station Cooperation Networks}, 
	year={2016},
	month={April},
	volume={65},
	number={4},
	pages={2729-2735},}

@ARTICLE{2025SFSA-LSJ,
  author={Liu, Shengjie and Yang, Chenyang and Han, Shengqian},
  journal={IEEE Trans. Wireless Commun.}, 
  title={Learning Wideband User Scheduling and Hybrid Precoding with Graph Neural Networks}, 
  year={early access, 2025},
  volume={},
  number={},}

@article{RandomEquivalent,
  author       = { Abboud, Ralph and
                  Ceylan, {\.I}smail{\.I}lkan and
                  Grohe, Martin and
                  Lukasiewicz, Thomas },
  title        = {The Surprising Power of Graph Neural Networks with Random Node Initialization},
  year         = {2021},
  journal       = {International Joint Conference on Artificial Intelligence},
}

@ARTICLE{ExpressiveSurvey2025,
  author={Zhang, Bingxu and Fan, Changjun and Liu, Shixuan and Huang, Kuihua and Zhao, Xiang and Huang, Jincai and Liu, Zhong},
  journal={IEEE Trans. Knowl. Data Eng.}, 
  title={The Expressive Power of Graph Neural Networks: A Survey}, 
  year={2025},
  volume={37},
  number={3},
  pages={1455-1474},
  month={March},}

@ARTICLE{2024GC_PY,
	author={Peng, Yao and Liu, Tingting and Yang, Chenyang},
	journal={IEEE GLOBECOM}, 
	title={Learning Power Allocation for Cell-free Massive
	{MIMO} System with Graph Neural Networks}, 
	year={2024},
	volume={},
	number={},}

@ARTICLE{2024PAGNN-Mishra,
  author={Mishra, Shashwat and Salaun, Lou and Yang, Hong and Chen, ChungShue},
  journal={IEEE Trans. Wireless Commun.}, 
  title={Graph Neural Network Aided Power Control in Partially Connected Cell-Free Massive {MIMO}}, 
  year={2024},
  volume={23},
  number={9},
  month={Sept.},
  pages={12412-12423},}

@ARTICLE{2023PAFNN-Fab,
  author={Fabiani, Mattia and Abdallah, Asmaa and Celik, Abdulkadir and Eltawil, AhmedM.},
  journal={IEEE GLOBECOM}, 
  title={Unsupervised Learning-Based Downlink Power Allocation for {CF-mMIMO} Networks}, 
  year={2023},}

@ARTICLE{2022pilotRDL-Rahmani,
  author={Rahmani, Mostafa and Dehghani, MohammadJavad and Xiao, Pei and Bashar, Manijeh and Debbah, Mérouane},
  journal={IEEE Access}, 
  title={Multi-Agent Reinforcement Learning-Based Pilot Assignment for Cell-Free Massive {MIMO} Systems}, 
  year={2022},
  volume={10},
  number={},
  pages={120492-120502},}

@ARTICLE{2024JPCPA-CF-Khan,
  author={Khan, MuhammadUsman and Testi, Enrico and Chiani, Marco and Paolini, Enrico},
  journal={IEEE PIMRC}, 
  title={Optimizing Power Control and Pilot Allocation in Cell-Free Massive {MIMO} via Deep Learning}, 
  year={2024},}

@article{2024JPCPA-CF-Ren,
	title={A Sequential Max {K}-Cut Approach for Pilot Assignment in Cell-Free Networks},
	author={ Ren, Boxiang and Hao, Han and Deng, Chaowen and Wang, Junyuan and Wu, Hao},
	journal={IEEE GLOBECOM},
	year={2024},
}

@ARTICLE{2025ZJYDesignGNN,
  author={Zhao, Jianyu and Yang, Chenyang and Liu, Tingting and Han, Shuangfeng and Wang, Xiaoyun},
  journal={IEEE Open J. Commun. Soc.}, 
  title={Designing Heterogeneous {GNNs} with Desired Permutation Properties for Wireless Resource Allocation}, 
  year={early access, 2025},
  volume={},
  number={},
  pages={},}

@ARTICLE{2024GNN_or_CNN,
	author={Zhao, Baichuan and Guo, Jia and Yang, Chenyang},
	journal={IEEE Trans. Commun.}, 
	title={Understanding the Performance of Learning Precoding Policies With Graph and Convolutional Neural Networks}, 
	year={2024},
	volume={72},
	number={9},
    month={Sept.},
	pages={5657-5673},}

@ARTICLE{2025MTS,
	title={A Tutorial on Multi-time Scale Optimization Models and Algorithms}, 
	author={Ramanujam, Asha and Li, Can},
	year={2025},
	journal={arXiv:2502.20568},
}

@ARTICLE{DataRateModel_GJ2023,
	author={Guo, Jia and Yang, Chenyang},
	journal={IEEE Trans. Commun.}, 
	title={Deep Neural Networks With Data Rate Model: Learning Power Allocation Efficiently}, 
	year={2023},
        month={March},
	volume={71},
	number={3},
	pages={1447-1461},}
\bibliographystyle{IEEEtran}

\end{document}